\pdfoutput=1
\documentclass[journal,comsoc]{IEEEtran}
\usepackage[T1]{fontenc}

\usepackage{multirow}
\usepackage{amsmath,amsthm, bm,amssymb,amsfonts,bbm}
\newtheorem{theorem}{Proposition}
\usepackage{graphicx,epsfig,subfigure}
\usepackage{algorithm}
\usepackage{algorithmicx}
\usepackage{fancybox}

\usepackage{cite,comment}
\usepackage{xcolor}
\usepackage{autobreak}
\usepackage{enumitem}
\usepackage{mathrsfs}
\usepackage[marginal]{footmisc}
\usepackage{algpseudocode}
\usepackage{tabularx}
\usepackage{booktabs}
\usepackage[labelfont=bf,format=plain,justification=raggedright,singlelinecheck=false]{caption}
\ifCLASSINFOpdf

\else

\fi

\usepackage{amsmath}

\usepackage[cmintegrals]{newtxmath}
\usepackage{hyperref}
\hypersetup{
    colorlinks=true,
    linkcolor=blue,
    filecolor=blue,
    urlcolor=cyan,
    citecolor=cyan,
}

\begin{document}

\title{A Trust-Centric Privacy-Preserving Blockchain for Dynamic Spectrum Management in IoT Networks}

\author{
	\IEEEauthorblockN{Jingwei~Ye,~Xin~Kang,~\IEEEmembership{Senior~Member,~IEEE},~Ying-Chang Liang,~\IEEEmembership{Fellow,~IEEE},~and~Sumei Sun,~\IEEEmembership{Fellow,~IEEE} }
	
	\thanks{J. Ye, X. Kang and Y.-C. Liang are with the Center for Intelligent Networking and Communications (CINC), University of Electronic Science and Technology of China (UESTC), Chengdu 611731, China (e-mail: {jingweiye97@gmail.com; ~kangxin@uestc.edu.cn; ~liangyc@ieee.org}).}
	\thanks{S. Sun is with the Institute for Infocomm Research, Agency for Science, Technology and Research, Singapore (e-mail: {sunsm@i2r.a-star.edu.sg}).}
}

\maketitle
\begin{abstract}
In this paper, we propose a trust-centric privacy-preserving blockchain for dynamic spectrum access in IoT networks.  
To be specific, we propose a trust evaluation mechanism to evaluate the trustworthiness of sensing nodes and design a Proof-of-Trust (PoT) consensus mechanism to build a scalable blockchain with high transaction-per-second (TPS). 
Moreover, a privacy protection scheme is proposed to protect sensors' real-time geolocatioin information when they upload sensing data to the blockchain.  
Two smart contracts are designed to make the whole procedure (spectrum sensing, spectrum auction, and spectrum allocation) run automatically.   
Simulation results demonstrate the expected computation cost of the PoT consensus algorithm for reliable sensing nodes is low, and the cooperative sensing performance is improved with the help of trust value evaluation mechanism. \looseness=-1
In addition, incentivization and security are also analyzed, which show that our design not only can encourage nodes' participation, but also resist to  many kinds of attacks which are frequently encountered in trust-based blockchain systems. 
\end{abstract}

\section{Introduction}

\IEEEPARstart{R}{ecent} years have witnessed the exponential growth of mobile data traffic. The rapid development of new wireless applications such as autonomous vehicles, remote healthcare, augmented and virtual reality will continue driving the demand for radio spectrum\cite{Cisco}.
However, the current static spectrum management method has resulted in under-utilization of spectrum resources\cite{FCC}.
To this end, dynamic spectrum access (DSA) has been proposed to  allow secondary users (SUs) sense and then use the idle spectrum bands. Since single SU's sensing capability is usually limited, cooperative sensing\cite{Akyildiz_CSS_Survey} is proposed to fuse the sensing results of multiple SUs to improve the sensing accuracy. However, the traditional usage of a centre to collect and fuse sensing results faces the risk of single point of failure and the increasing management and regulatory costs. Moreover, SUs usually need to upload their sensing data to the fusion center at the risk of leakage of their privacy, such as identity and location.
\par As an emerging technology, blockchain has attracted attention from both academia and industry.
By using blockchain, a decentralized resource management system can be established without the requirement of trustworthy central management agencies\cite{Access_Crowdsensing_Applications, TCOM_BC_Application, B4SDC, MCSChain, crowdbc}. Moreover, smart contracts which are implemented on a blockchain can be used to replace traditional central management agencies and facilitate the cooperation of users \cite{YANG2019408}. Recently, researchers have investigated how to apply the blockchain technology to DSA. \cite{TCCN} summarizes the types of blockchains applicable in different spectrum sharing scenarios, and the possible advantages and disadvantages of applying blockchain technology to DSA. With the help of cryptocurrency issued by the blockchain and  smart contracts, flexible and automatic spectrum trading markets are made possible for spectrum sellers and buyers. \cite{DBLP:journals/corr/PascaleMMD17,8761664,8851203,8885750,Spass}.
\par However, the integration of blockchain and DSA still faces many challenges.
Firstly, although the data recorded on a blockchain is prevented from being tampered, the quality or value of it cannot be guaranteed. Especially, in a public blockchain, a malicious user can easily join the blockchain and record their data in the blockchain. The data from such a malicious user or an unreliable user is valueless or even harmful to cooperation  tasks based on the data, e.g., cooperative sensing. Therefore, there is compelling need to evaluate the quality of the data from each user.
Secondly, though the communication among Ethereum external accounts can be protected by encryption and decryption with their public and private key pairs, a contract account is not equipped with
a key pair hence such protection method is not applicable to smart contracts. As a result, sensing results uploaded to the smart contract cannot be encrypted, and the sensitive information in the sensing results (such as real-time geolocatioin), can thus be accessed by malicious users.
Thirdly, traditional consensus algorithms like Proof of Work (PoW) used in public blockchain introduce much computation overhead and thus make the transaction processing speed very low. As a result, it is inefficient to directly use existing consensus algorithms (such as PoW) for spectrum management, especially considering the limited computation capabilities of IoT devices.
\par In this paper, we consider the DSA for an IoT network where IoT devices opportunistically access the licensed spectrum bands through cooperative sensing. Blockchain is used as a platform for dynamic spectrum management. Specifically, key component designs for such a blockchain-enabled DSA system, including the trust evaluation mechanism, privacy protection, consensus algorithm and smart contracts, are studied. The main contributions of this paper are summarized as follows.
\begin{itemize}
    \item  We propose a trust evaluation mechanism to evaluate the trustworthiness of sensing nodes which participate in the collaborative sensing. We show that the proposed trust evaluation mechanism is effective in incentivizing sensing nodes to be honest, and thus improving the spectrum sensing result. 
    \item We propose a strong privacy protection mechanism by combining ring signature with the commit-and-reveal scheme for providing sensing nodes' privacy such as their real-time location. The ring signature is used to protect sensing nodes' identities when they upload sensing data. The commit-and-reveal scheme is proposed to figure out  the connection between sensing data packets and its originated sensing node after the fusion of sensing results, so that the trust value of each sensing node is updated and the incentive tokens are distributed.
    \item We propose a new consensus algorithm named as Proof-of-Trust (PoT) by connecting the mining difficulty with a miner's trust value. We show that the proposed PoT can greatly reduce the computation cost of honest nodes and enhance the scalability of current blockchian-based DSA system. 
    \item We design the system protocol and two smart contracts to make the whole DSA procedure, including spectrum sensing, spectrum auction, and spectrum allocation, run automatically. In addition, we implement and verify a prototype of our proposed protocol and smart contracts using Solidity\cite{Solidity} on the Ethereum test net.
\end{itemize}
\par The rest of the paper is organized as follows. In Section \ref{sec:preliminaries}, we introduce the basic concept of blockchain and smart contract. In Section \ref{sec:Related work}, we describe the works related to the application of blockchain to DSA. In Section \ref{sec: system}, we introduce our system model. In Section \ref{sec:Proposed protocol}, we describe the design of our proposed protocol including the block structure, the trust evaluation mechanism, the consensus algorithm and the privacy protection scheme. In section \ref{sec:Protocol_and_implementation}, the design of smart contract and workflow are introduced. In Section \ref{sec:System analysis}, we discuss the performance of the trust evaluation mechanism, and analyze the security of the system. Finally, Section \ref{sec:conclusion} concludes this paper. \looseness=-1

\section{Preliminaries} \label{sec:preliminaries}
\subsection{Blockchain}
\par Blockchain is a growing chain of data blocks with a specific data structure constructed using cryptographic algorithms. It uses hash pointers to connect the data blocks to form an entangled chain. In this way, the integrity of the data can be protected. Blockchain uses a measurable and verifiable mechanism to reach a consensus among nodes in the network on the generation of new blocks. Such a mechanism is usually referred to as the \textit{consensus algorithm}. On the premise that the data block entanglement is guaranteed, the structure and the consensus algorithm of the blockchain can be flexibly designed based on the demand of application scenarios.
\par Blockchain can be roughly classified into three types, depending on degree of openness, which are public blockchain, private blockchain, and consortium blockchain. Public blockchain is designed to enable and to record the transactions in a public network, which means that any node can freely join or leave the blockchain network without authorization. Bitcoin, the most famous digital cryptocurrency, is generated and managed by a public blockchain. Consortium blockchain is a permissioned blockchain, which is open to authorized members of external institutions in limited roles and functions. The authorization of the node to join are all determined by an authorized organization. The nodes in a private blockchain network  trust each other. Thus, private chains can simplify operations in data authentication and the consensus algorithm to improve the efficiency.
\subsection{Smart Contract}
In the 1990s, Nick Szabo first proposed the concept of smart contracts\cite{Szabo_1997}. It was defined as computerized transaction protocols that execute terms of a contract. In the field of cryptocurrency, a smart contract is defined as an application or a program that runs on a blockchain. Normally, it contains a set of digital agreements with specific rules. These rules are predefined in the form of computer source codes, and all network nodes will copy and execute these computer source codes independently. Smart contracts are highly customizable and can be flexibly designed to provide different services and solutions.

\subsection{Ring Signature}
Ring signatures were first introduced in 2001\cite{RSIG} which follow a ring-like structure of the signature algorithm. It is a type of digital signature that can be performed by any member of a group of users without the agreement of others. Therefore, the only ``open'' information is that a message signed with a ring signature is endorsed by someone in a particular group of people. Ring signatures are deliberately designed so that it is computationally infeasible to determine which of the group members' keys was used to produce the signature. Ring signatures are similar to group signatures but differ in two key ways: firstly, there is no group administrator to revoke the anonymity of an individual signature; secondly, any set of users can be used as a group by the signer without additional setup.

\section{Related work} \label{sec:Related work}

\par Blockchain was first proposed as a secure decentralized ledger for recording spectrum transaction data and analyzing overall spectrum usage in dynamic spectrum management.
In \cite{Kotobi}, the authors proposed to use the blockchain as a distributed database to securely store the transactions in dynamic spectrum sharing, and to use the digital currency issued by the blockchain as the currency for spectrum auctions. In \cite{Spass}, a smart contract running on Ethereum was proposed as a platform to provide spectrum sensing services.

In \cite{8851203}, a blockchain-based spectrum trading and sharing scheme was proposed for Mobile Network Operators (MNOs) to lease their spectrum to secondary Unmanned Aerial Vehicle (UAV).
In \cite{axiveSC}, spectrum sharing contracts deployed on a permissioned blockchain platform are constructed for multi-operators spectrum trading.
In \cite{8885750}, the authors proposeds to use the smart contract to securely store and then autonomously implement the spectrum leasing agreements, thus to avoid the interference from the disordered spectrum access.
\par However, these works did not make any improvement to the existing blockchain.  Thus, the TPS of the blockchain is low, limiting the performance of the blockchain-based DSA platform. Thus, 

in \cite{KM-BC}, the authors proposed a consensus algorithm named as "blockchain-KM protocol" by modifying the PoW, to speed up the transaction processing for the license-free spectrum sharing.
In \cite{Proof-of-Strategy}, a novel consensus algorithm named as proof-of-strategy was proposed, where the best strategy to allocate the spectrum bands in the whole network is the proof of authority to publish a new block. Such a consensus algorithm can not only regulate the transaction generation, but also optimize the spectrum allocation.

In \cite{Proof-of-Device}, by modifying the PoW, the authors proposed a consensus algorithm named as Proof-of-Device, which selects a device, represented by its unique identifier, to publish a new block in a way like lottery. It thus eliminates the heavy computation cost in PoW.

On the other hand, blockchain system is an open system, hence the reliability of data source needs to be guaranteed. 
In \cite{senseChain}, the data uploaded by the sensors are compared with the data collected by the trustworthy validator to determine the credibility degree of the sensor, in this way, the system can evaluate the credibility of nodes for better performance. However, this method is a semi-centralized method which needs to deploy trustworthy validators.

Different from the above works, we not only delve into the design of block structure and consensus algorithm, intending to overcome the scalability problem of the blockchain technology, but also design a decentralized trust evaluation mechanism and a privacy protection mechanism for the openness and transparency of blockchain system while protecting the data's confidentiality and the sensors' privacy.
To the best of our knowledge, this is the first work of joint trust evaluation and privacy protection mechanism in cooperative sensing in a blockchain-enabled DSA system.

\section{System Model, challenges and design guidelines} \label{sec: system}
In this section, we introduce the dynamic spectrum sharing scenario in the future IoT network, and list the difficulties and challenges in applying the blockchain to this scenario. Finally, we give the guidelines on how to design a blockchain-based DSA system.
\subsection{Blockchain-enabled DSA System Model}\label{subsec:scenario}

In this paper, we consider the DSA in a two-tier cognitive radio network consisting of licensed/primary users (PUs) and unlicensed/secondary users (SUs). In this network, PUs are inactive users who may not use their spectrum all the time so SUs can detect these spectrum holes and use them. As illustrated in Fig.2, SUs can be a variety of heterogeneous IoT devices, which form a mesh network. SUs access the spectrum bands of a PU in an opportunistic manner based on cooperative sensing, and a blockchain is used as the platform of the cooperative spectrum sensing decision and the spectrum allocation.
\begin{figure}[htbp]
\centering
\includegraphics[width=3in]{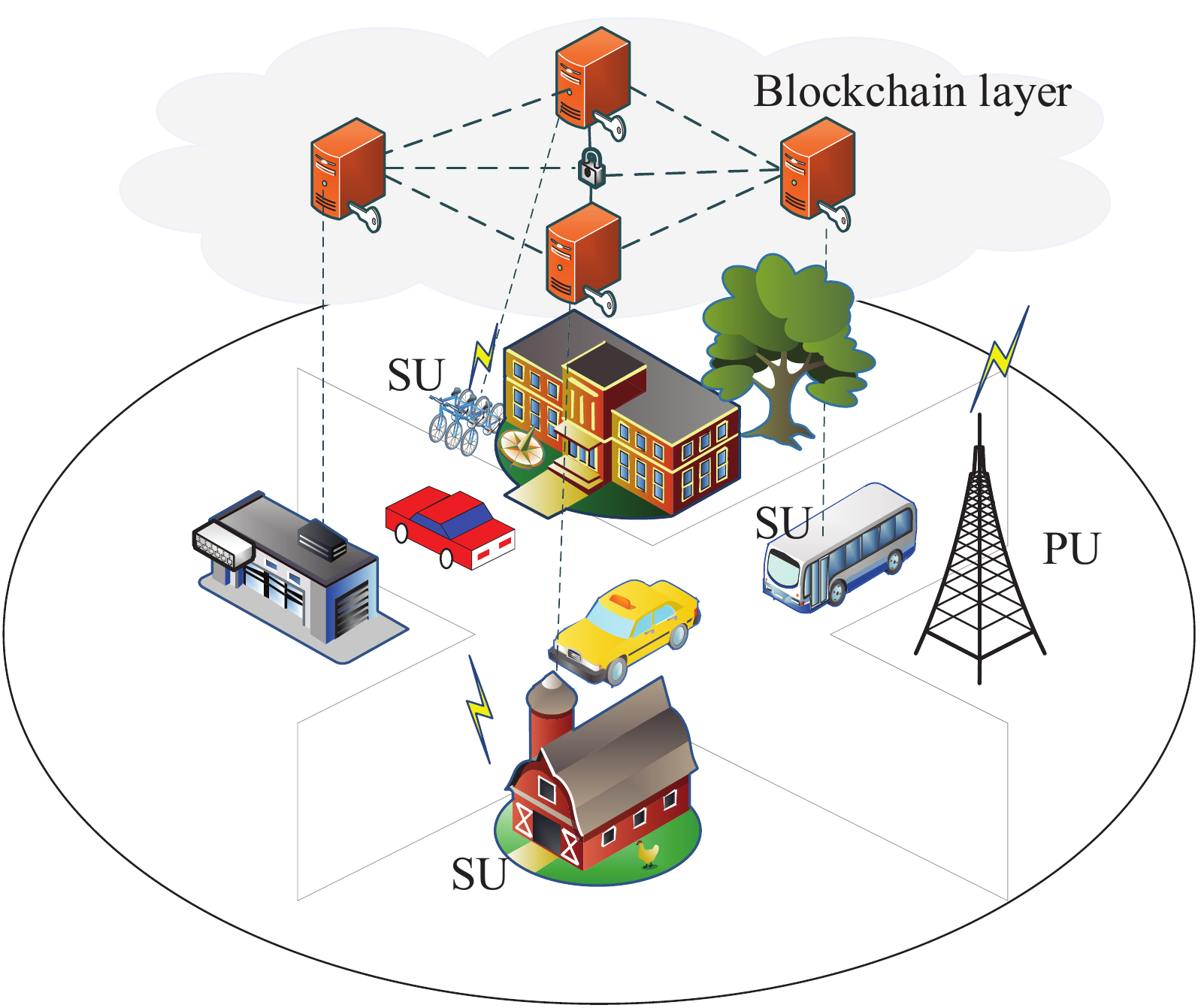}
\caption{System model} \label{CRnetwork}
\end{figure}

Basically, as shown in Fig. 3, there are four phases in our blockchain-enabled DSA system.
\begin{enumerate}
\item \textit{Individual Sensing}: The sensors perform spectrum sensing to detect the state of the PU's spectrum band.
\item \textit{Sensing Fusion}: Each of the sensor broadcasts its sensing data which mainly consists of the sensing result and the geolocation where the sensing is performed. These data can be raw data or quantized bits, which will be fed to a smart contract implemented upon a blockchain. Then, the smart contract provides the final sensing result according to predefined fusion rules.
\item \textit{Spectrum Allocation}: When the final sensing result indicates that PU is inactive, it should be decided which spectrum requester can access the spectrum. Spectrum auction, as a common and fair way for spectrum allocation, is chosen to achieve that. In the spectrum auction, the digital cryptocurrency, which is issued by the blockchain through incentive mechanisms for spectrum sensing and mining, is commonly considered.
\item \textit{Spectrum Access}: The SU who obtains the access right through spectrum auction accesses the idle spectrum for data transmission.
\end{enumerate}

\par Compared with traditional cooperative spectrum sensing and spectrum allocation methods, blockchain-enabled DSA has the following benefits:
\begin{itemize}
\item \textit{Decentralized}: Blockchain-based dynamic spectrum access does not require the deployment of trusted central nodes which avoids the single point of failure.
\item \textit{Transparency}: The blockchain ledger can record the whole process of DSA.
\item \textit{Automatic}: Using smart contract instead of traditional contracts, we can achieve automatic spectrum management and on-chain payment settlement.
\item \textit{Flexibility}: For spectrum trading on the blockchain platform, diversified spectrum trading rules can be dynamically enforced by adjusting parameters of smart contracts. Compared with the current fixed paper-based contracts signed between mobile users and operators, this method has higher flexibility.
\end{itemize}
\begin{figure}[htbp]
	\centering
	\includegraphics[width=3in]{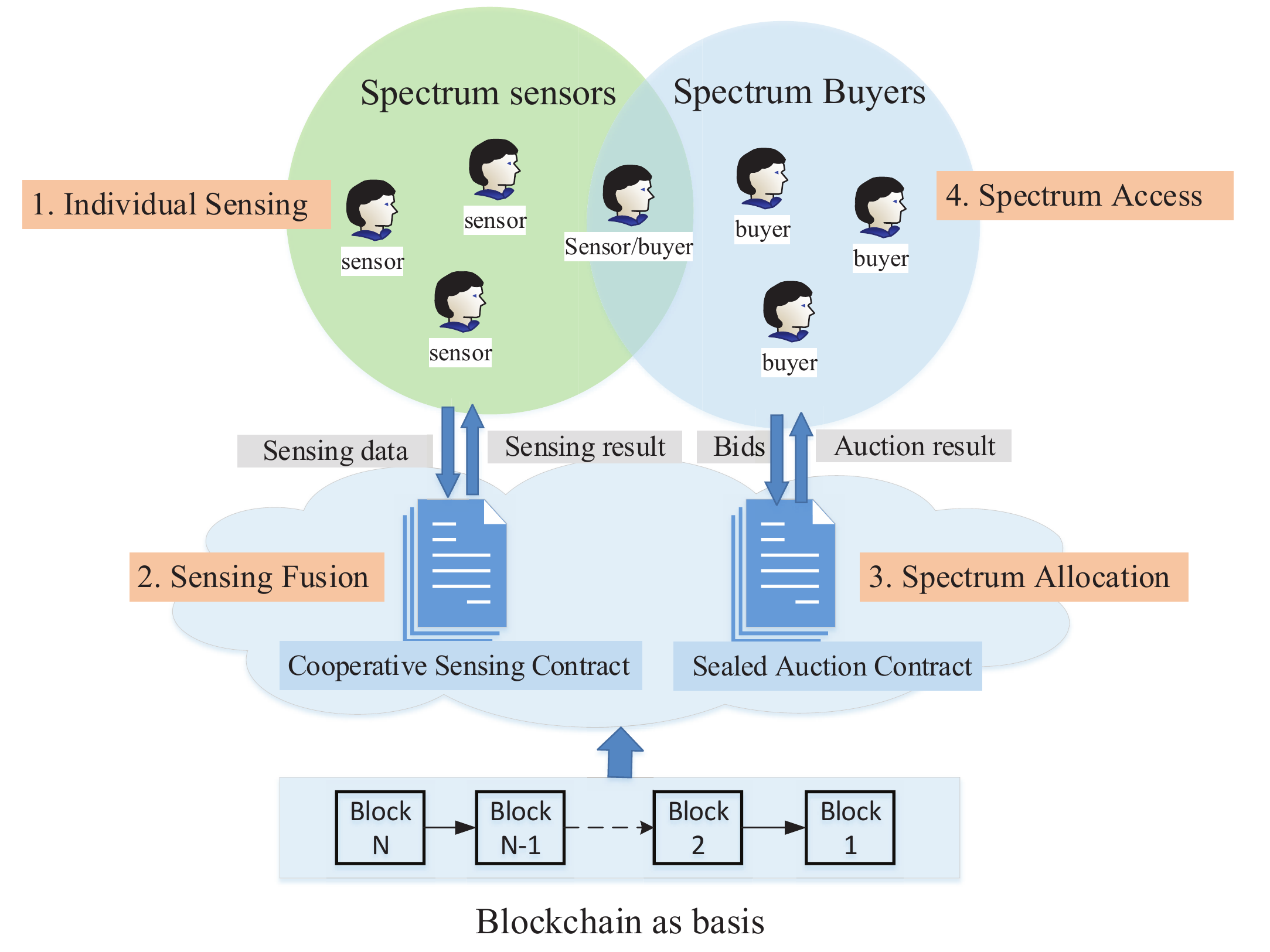}
	\caption{Blockchain based DSA}\label{BCBasedDSA}
\end{figure}

\subsection{Difficulties and Challenges}\label{subsec:challenge}
\par Traditional cooperative spectrum sensing is divided into two modes: centralized mode and distributed mode.
Centralized collaborative spectrum sensing requires the deployment of a central fusion node to process the data collected from sensors which usually contain sensors' geolocation information. 
Therefore, this mode not only requires redundant regulation fees (usually charged to maintain the operation of the central node), but also faces the risk of data leakage when the central node is attacked.
While the distributed mode often suffers from the fake reporting issue caused by malicious nodes and the iterative-based malicious node detection algorithms \cite{TCCN-Trust} introduce extra computation cost.
\par Therefore, how to implement the cooperative sensing in a secure and effective way is a challenge.
The decentralized nature of the blockchain system is promising to realize distributed cooperative sensing and encourage all sensors to share spectrum sensing results. However, there are still some challenges when applying blockchain into spectrum management.
Firstly, IoT devices usually lack sufficient computing and storage resources to maintain and store the whole blockchain ledger.
Secondly, a malicious node can easily join a public blockchain network and threaten the reliability of the cooperative spectrum sensing.
Moreover, the sensing-related data recorded on blockchain or collected by smart contracts is usually in an unencrypted format which may leak the private information of IoT devices contained in their sensing data packets.

\subsection{Design Guidelines}
Considering above challenges, our discussion on system design mainly includes the following four parts.
\begin{itemize}
    \item \emph{Trust Evaluation}: Specially, in a public blockchain network, anyone can join or exit the system at will, and one user can create multiple accounts in the system. Therefore, if there is no suitable node evaluation mechanism, malicious users can launch sybil attacks easily. This is because, malicious nodes can keep on applying for accounts and continue launching attacks. Therefore, the system should be able to evaluate the trust values for nodes, and identity the malicious nodes based on the their trust values. 
    \item \emph{Lightweight Consensus Algorithm}: Due to the limitation of the computing power and storage resources of IoT devices, maintaining a normal public blockchain system is a heavy burden. To reduce the computing cost, it is necessary to invent a new consensus algorithm with low computational and storage complexity.
    \par On the other hand, as time goes by, the length of the ledger keeps on increasing, resulting in increased node storage. To reduce the storage cost, methods such as edge storage can alleviate storage pressure on the blockchain nodes but may introduce additional budget. Thus, a method that can directly reduce the storage cost of the blockchain is preferred.
    \item \emph{Privacy Protection Mechanism}: Smart contracts does not have public-and-private key pair and thus can not communicate with external accounts in an encrypted manner. Therefore, when a smart contract is used as the data fusion center of cooperative spectrum sensing, a privacy protection mechanism is needed to ensure that the private information of the sensors will not be leaked during the interaction between sensors and the smart contract.

\end{itemize}

\section{Blockchain Design, Trust Evaluation and Privacy Protection} \label{sec:Proposed protocol}
In this section, to improve the accuracy of cooperative sensing, 
we first introduce the block structure of our blockchain.
We then propose the trust evaluation mechanism to evaluate the trustworthiness of each SU on spectrum sensing.
After that, a high-efficiency consensus algorithm is proposed for our blockchain.
Finally, a privacy protection mechanism is proposed for protecting participants' privacy.

\begin{table}[h!t]
\scriptsize
\caption{List of notations  }
\label{tab4}
\begin{tabular}{|l|l|}
\hline
Notations & Descriptions \\
\hline
  $N_{i,w}(n)$ & Accumulated number of wrong sensing results of $sensor_i$ at round $n$ \\
  $N_{i,r}(n)$ & Accumulated number of right sensing results of $sensor_i$ at round $n$ \\
  $P $ & Length of window for calculating $N_{i,w}(n)$ \\
  $\mathcal{D}^i_n $ & Mining difficulty of node $i$ at block $n$\\
  $\mathcal{H(\cdot)} $ & Hash operation \\
  $R_{sleep} $ & Number of inactive rounds since last active round \\
  $RND_i $ & Random number picked by $sensor_i$ \\
  $SR$ & Sensing result of unknown but legal sensors \\
  $SR_i $ & Sensing result of $sensor_i$ \\
  $N_1$ & Number of required sensors in Cooperative Sensing Contract ($CSC$) \\
  $N_2$ & Number of required bidders in Spectrum Auction Contract ($SAC$) \\
  $TV_{thr}$ & Trust value threshold set in $CSC$ \\
  $T_{ddl}$ & Sensing deadline in $SAC$ \\
  $d_s$ & Deposit for sensing \\
  $d_a$ & Deposit for auction \\
  $T_{self-d}$ & Time for $SAC$ to self-destruct \\
  $TS$ & Timestamp \\
  $loc$ & Location of unknown but legal sensors  \\
\hline
\end{tabular}
\end{table}

\subsection{Block Structure}\label{subsec:BS}
The block structure should be tailor-designed for the blockchain-based DSA system. Some of the unnecessary components in the block header can be removed to reduce the overheads on the block transmission and storage and other necessary information, such as the trust value of the block generator, needs to be recorded in the block header to prevent possible attacks in the DSA system.
\begin{itemize}
\item Block header: As shown in Fig. \ref{Pic:BlockStructure}, The block header mainly contains (i) the merkle root of transactions; (ii) the merkle root of account states; (iii) the hash value of the header of the previous block, as the hash pointer of the chain structure; (iv) the trust value of the miner who mines this block; (v) the signature of the miner who mines this block; (vi) a \textit{Nonce}, which is the solution of Hash puzzle of the consensus algorithm; (vii) the timestamp when the block is mined.
\item Block body: The block body mainly stores transaction generated during spectrum allocation using merkle tree. Unlike the unspent transaction output (UTXO) \cite{UTXO} model used in Bitcoin, the system based on the account model usually contains an account tree to record the account states such as the balance and the trust value of each account. Merkle Patricia Trie (MPT) \cite{wood2014ethereum} which is a commonly adopted data structure, is used here to store the account states.
\end{itemize}

\begin{figure}[htbp]
\centering
\includegraphics[width=3in]{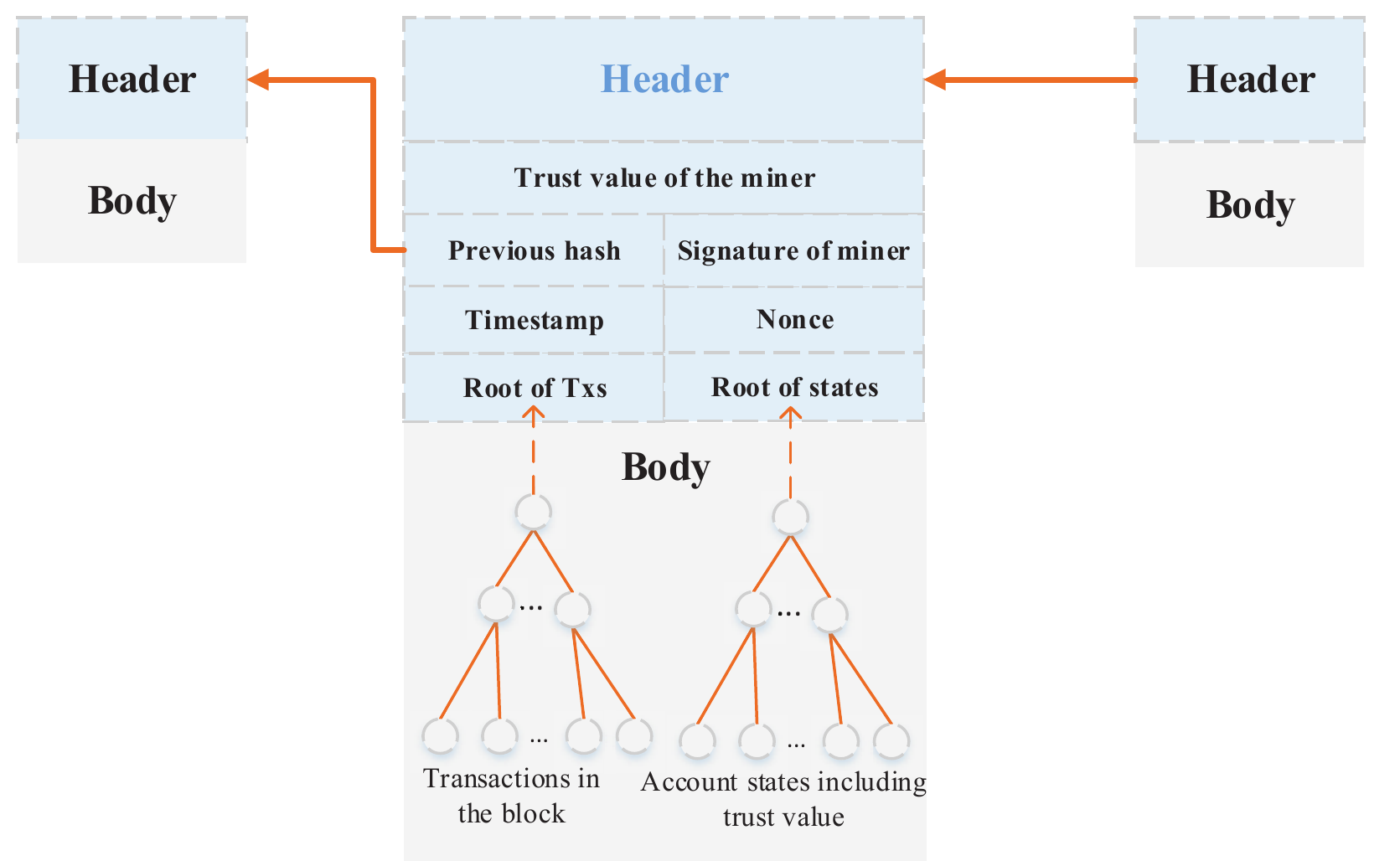}
\caption{Block structure} \label{Pic:BlockStructure}
\end{figure}
\subsection{Trust Evaluation Mechanism}\label{subsec:TV design}
Although the blockchain can guarantee that the data recorded in a block will not be tampered, it cannot guarantee that the data source is trustworthy.
Especially, in a public blockchain, a malicious node can upload misleading data which may interfere the cooperation sensing based on such data. In such cases, the performance of the cooperative sensing will be degraded, and the honest and reliable sensors will be discouraged from participating in cooperative sensing. To this end, we design a trust evaluation mechanism to evaluate the credibility of every sensing node in the blockchain.
Instead of using a central authority to record and maintain every account's trust value, we propose to include the trust value as an attribute of the blockchain account, which is recorded in the block and maintained by all the miners. In the following, we will discuss the design of trust value updating method.

\par Intuitively, the trust value of a sensing node should be adjusted based on its performance on the current and historical cooperative sensing. In detail, if it plays a positive role in the cooperative sensing, its trust value should be increased, and vice versa. Here, we define that the effect of a sensor is \textit{positive} when its sensing result is consistent with the cooperative sensing result, and is \textit{negative} when its sensing result is inconsistent with the cooperative sensing result, or when the sensing node does not participate in the cooperative sensing.
Denote the number of times when the sensing result of the \textit{i}-th node is consistent with the cooperative sensing result as $N_{i,r}$; the number of times when the sensing result of a node is inconsistent with the cooperative sensing result as $N_{i,w}$.
Inspired by \cite{KangTV}, we propose a new model to calculate trust value of $sensor_i$ denoted by $TV_i^{(1)}(n)$ in sensing round $n$, where $n$ is also the number of sensing tasks published since this system was created.
        \begin{equation}
        \begin{aligned}
        TV_{i}^{(1)}(n) = e^{-\rho  N_{i,w}(n)}\left(1 - e^{-\eta  N_{i,r}(n)}\right),
        \end{aligned}\label{con:tvc}
        \end{equation}
where $\rho>0$ and $\eta>0$ are coefficient that determine how fast the trust value changes with respect to $N_{i,w}$ and $N_{i,r}$ , respectively.
\par Moreover, we let $N_{i,w}(n)$  decay by $\frac{1}{P}$ every time when the node participates in sensing, and we only consider the latest $P$ sensing rounds of each node. In this way, the effect of a wrong sensing result on the trust value will gradually be degraded over time. As a result, a node that unintentionally submitting an inconsistent result will be gradually forgotten. Mathematically, $N_{i,w}(n)$ can be defined as 
\begin{equation}\label{Nwrong}
  N_{i,w}(n) = \sum\limits_{m=n-P}^{n} r_{i}(m)\left(1-\frac{n-m}{P}\right),
\end{equation}
where $r_{i}(m) = 1$ if $sensor_i$ broadcasts the wrong sensing result in round $m$, and $r_{i}(m) = 0$ otherwise.
\par Besides, we denote the number of rounds when a sensing node is inactive as $R_{sleep}$, and we add a function $f(R_{sleep})$ which models the negative impact of being inactive on the trust value. It is designed to satisfy the following criteria. 
\begin{itemize}
	\item ${\forall}R_{sleep} \in [1,+\infty)$, $f(R_{sleep}) \in (0,1]$,
	\item $\frac{\partial f_{dcv}}{\partial R_{sleep}} <0$,
\end{itemize}
\par The first criterion is for normalization; the second criterion is to ensure that the negative effect increases as  $R_{sleep}$ increases. 
Moreover, as $R_{sleep}$ increases, the downward trend of the function is first gentle, and then becomes severe. This is designed to punish the peers that does not participate in the sensing process. The punishment is light when the peer does not participate only a few times, but becomes severe when the node always does not participate in the sensing. $f(R_{sleep})$ will reach a certain low value when $R_{sleep}$ is big enough.
For example, the trust value of a node will reduce to $r_1$ ($r_1<1$) of its original value for $k_1$ consecutive non-participation, and to $r_2$ ($r_2<r_1<1$) for $k_2$ consecutive non-participation.
Here, a piecewise function consisting of multiple linear functions is utilized as our $f(R_{sleep})$.
\begin{equation}
	f(R_{sleep})=\left\{
	\begin{array}{lcl}
		\frac{k_1-R_{sleep}}{k_1}\cdot (1-r_1) + r_1 && {0 \leq R_{sleep} \leq k_1}\\
		\\
		\frac{R_{sleep}-k_2}{k_1-k_2}\cdot (r_1-r_2) + r_2  && {k_1 \leq R_{sleep} \leq k_2}\\
		\\
		r_2 && {k_2 < R_{sleep}}\\
	\end{array} \right.
\end{equation} 

These parameters can be set flexibly. If the actual deployment requires stricter penalties, then all these parameters should be set smaller.

\par By adding the attenuation function, the model in (\ref{con:tvc}) can be modified as:
\begin{equation}
TV^{(2)}_i(n)=
\begin{cases}
TV_{i}^{(2)}(n-1)+ \Delta TV_{i} \cdot f(R_{sleep}),&\mbox{if $\Delta TV_{i}>0$}\\  \\
TV_{i}^{(1)}(n),&\mbox{if $\Delta TV_{i}<0$}\\  \\
TV_{i}^{(2)}(n-1)\cdot f(R_{sleep}),&\mbox{if inactive}
\end{cases}\label{con:tvc2}
\end{equation}  where $\Delta TV_{i}$ = $TV_{i}^{(1)}(n)-TV_{i}^{(2)}(n-1)$ is defined as the original increment of trust value, which can be either positive or negative.
The design criterion of (\ref{con:tvc2}) is explained as follows. We consider three cases. First of all, when the sensor senses correctly, the original increment of trust value is positive, i.e.,  $\Delta TV_{i}>0$. Such increment is first degraded by the negative effect of sleeps in the previous sensing activities, i.e., $f(R_{sleep})$, and then added to the trust value. Secondly, when the sensor gives the wrong sensing result, the original increment of trust value is negative, i.e., $\Delta TV_{i}<0$. In this case, we choose not to consider the negative effect of the sleep in previous sensing rounds, in order to encourage the sensors which unintentionally derive the wrong sensing result in this time. Finally, for sensors who do not participate in sensing, their trust value will not remain unchanged but gradually decrease as their sleep time increases. This design is to ensure that sensors can only maintain its high trust value by actively participating in spectrum sensing.

\subsection{Consensus Algorithm, Forking and Block Compression} \label{ subsec:Consensus algorithm}
\par Based on the last block of the block chain, nodes in the blockchain network will compete for publishing the next block, and only one block will be accepted as the next valid block.

The traditional consensus algorithm for public blockchain like PoW is computation-intensive, which may be inapplicable in IoT networks, where IoT nodes are usually with limited computational capabilities. One reason why the PoW is designed to be computation-insensitive is that it assumes there is nearly no trust among nodes in a blockchain network. Nevertheless, with the trust value mechanism proposed in this paper, we can evaluate the credibility of a node on spectrum sensing in the blockchain network. Therefore, based on the trust value, we can optimize our design of the blockchain consensus algorithm, forking solution and blockchain compression. 
\par \textit{Consensus  Algorithm:} Intuitively, the nodes with higher trust value on sensing are more likely to be honest in publishing a new block. Therefore, we can reduce the difficulty of such nodes to mine a new block. In this way, the computation consumption for reliable nodes will be decreased, and a block will thus be more quickly mined and published.

Based on PoW and trust value, we propose a Proof of Trust (PoT) consensus algorithm. Similar to the Proof-of-Stake (PoS), PoT assigns different nodes different mining difficulties. 
To be specific, for successful mining, the miner with the higher trust value is required to find a hash value with fewer leading zeros and vice versa. 

\par Mathematically, assuming that the difficulty of node $i$ at block $n$ is denoted as $\mathcal{D}^i_n$, which can be denoted as
\begin{equation}\label{equ:difficulty}
 \mathcal{D}^i_n= \beta_n \cdot \left(1 -sin\left(\frac{\pi}{2}\cdot TV_{i}\right)\right),
\end{equation}
where $\beta_n$ denotes the base difficulty of block $n$. $TV_i$ is the trust value of node $i$.  The initial difficulty is denoted as $\beta_{0}$. $\beta_{0}$ can be determined by evaluating the computing power of actual IoT devices. We will discuss this problem in the simulation section.
A node with $TV_i=0.8$ has about $0.049 \cdot\beta_n$ difficulty, which is 17 times less than a node with $TV_i=0.1$ whose difficulty is about $0.844 \cdot \beta_n$.
\par We denote the timestamp of block $n$ as $T_n$, the ideal time interval of two consecutive blocks is $T_0$. Then,we have
\begin{equation}
	q = \left \lfloor \frac{T_{n-1}-T_{n-2}}{T_0}\right \rfloor,
\end{equation}
where $q$ is defined as adaptive adjustment factor of base difficulty, and $g$ is defined as the difficulty adjustment granularity, which can be written as  
\begin{equation}
	g =\left \lfloor \frac{\beta_{n-1}}{128} \right \rfloor,
\end{equation}
The updating of $\beta_n$ can be denoted as 
\begin{equation}
	\beta_n = \beta_{n-1} - q\cdot g,
\end{equation}
 
\par \textit{Forking Solution:} As the mining speed increases, there may be multiple blocks mined at nearly the same time. Because of the communication latency, the block that is first mined might not be the first to be received by all the nodes in the blockchain network. Instead, the block that is first received and recognized by different nodes might be different. In this case, the blockchain \textit{forking} occurs. It is harmful and thus needs to be solved. In Algorithm \ref{alg1}, we propose a trust value based forking solution. It first compares the trust values in the head of blocks, and it then validates the block with the highest trust value. If the trust values of multiple blocks are the same, the block whose timestamp is earlier is selected as the valid block. If the timestamp is the same, it compares the hash value of the blocks, and then selects the block with the smallest hash value as the valid block. Since the probability of hash collision is negligible, Algorithm \ref{alg1} can select one valid block eventually.
\par \textit{Blockchain Compression:} Since IoT devices are usually with limited storage space, each time when the blockchain grows by $L$ blocks, compression will be performed. For the compression of blockchain, in the existing literature, the RSA accumulator as a data structure, which functions similarly to that of a Merkle tree, can be used to compress the blockchain \cite{boneh2019batching}. Another approach to compress blockchain is to use the chameleon hash function to replace the traditional hash function of the blockchain \cite{RedactableBC}. Here, using the trust value, we propose a compression method called as \textit{Trust-based Compression (TBC)}. The node with the highest current trust value will be authorized to compress the blockchain. To be specific, such a node will first extract the account tree in the last block as the body of the new block, calculate a hash value for the body, and finally combine the obtained nonce, miner's signature, and timestamp to form the block header. The obtained block is the new genesis block. Since each node stores the original blockchain, it is easy to check whether the account status has been tampered by the selected node during the compression process. After a node verifies the new genius block, it will clear the original blockchain.
\begin{algorithm}[t]
\caption{Block Selection}\label{alg1}
\hspace*{0.02in} {\bf Input:} 
the block to be compared: $Block_i$, $Block_j$;\\
\hspace*{0.02in} {\bf Output:} 
The winner block;
\begin{algorithmic}
\If{$TV_i < TV_j$}
    \State $Block_i$ wins.
\ElsIf{$TV_i > TV_j$} 
	\State $Block_j$ wins.
\Else
	\If{$Timestamp_i<Timestamp_j$}
    	\State $Block_i$ wins.
	\ElsIf{$Timestamp_i>Timestamp_j$}
		\State $Block_j$ wins.
	\Else
		\If{$\mathcal{H}(Block_i)<\mathcal{H}(Block_j)$}
    		\State $Block_i$ wins.
		\Else
    		\State $Block_j$ wins.
    	\EndIf
     \EndIf
\EndIf
\end{algorithmic}
\end{algorithm}

\subsection{Privacy Protection Mechanism}\label{Subsec:Privacy protection}

The location where a sensing node senses the spectrum bands of interest is useful for the fusion centre to cluster the sensing nodes and to improve the cooperative sensing accuracy \cite{locationFusion}. Therefore, alongside the sensing result, the location is needed to be uploaded by a sensing node. However, it is difficult to protect the location information of a sensing node from being leaked when a smart contract is used as the fusion centre. This is because a smart contract account is not equipped with a key pair which can be used by sensors to encrypt their upload data packet.
\par The privacy protection issue in this case is to hide the source of a sensing packet, i.e., from which sensing node the sensing packet comes. However, the sensing node cannot be allowed to be totally anonymous when submitting the sensing packet because this will make the cooperative sensing system vulnerable to malicious attacks. To this end, we propose the use of the ring signature \cite{RSIG} by each sensing node to hide the source of a sensing packet in a group of valid sensing nodes. Also, the smart contract as the fusion centre can identify the validity of each received sensing packet. Moreover, since sensing packets are unencrypted, fusing these sensing results can be carried out directly and automatically in the smart contract.
\par In the following, we illustrate the procedures of one sensor, denoted by $Sensor_s$, in generating the ring signature.
\begin{itemize}
\item[1.] $Sensor_s$ selects $n-1$ legal sensors to form a group and collect their public keys. The format of a sensing data packet, denoted as \textit{msg}, is given as follows 
\begin{equation}
	msg = \{\mathcal{H}(msgID), SR,time, location\}\label{equ:msg},
\end{equation} 
where $\mathcal{H}(msgID)$ is used in the update of the trust value and $SR$ is the sensing result denoted by one bit.
\item[2.] $Sensor_s$ uses one-way hash function to compute $k = \mathcal{H}(msg)$, which is a symmetric key of the symmetric encryption function $E_k$. 
\item[3.] $Sensor_s$ generates a random value for each of other members in the group where it belongs to. Specifically, the random number $x_i$ is first generated for the $i$-th node in the group. It then calculates corresponding $y_i$=$g_i(x_i)$ using corresponding public key, where the function $g_i(\cdot)$ is the encryption function encrypted with the public key $pk_i$.
\item[4.] $Sensor_s$ finds the solution to the ring equation (\ref{equ:ringEquation}) and gets the undetermined parameter $y_s$, where $v$ is a random value chosen by $Sensor_s$. Then, $Sensor_s$ calculates $x_s$ using its own privacy key: $x_s$=$g_{s}^{-1}(y_s)$. To find solution, a private key of a sensor in this group is needed , anyone who is not in the group cannot generate a legal ring signature.
    \begin{equation}
    \begin{split}
       C_{k,v}(y_1,y_2,...,y_n) & =E_k(y_n \oplus E_k(y_{n-1} \oplus E_k(...  \\
         & \oplus E_k(y_1\oplus v)...)))=v,
    \end{split}\label{equ:ringEquation}
    \end{equation}
\end{itemize}
\par Finally, a valid ring signature, denoted as $Ring_{sig}$, can be generated by $Sensor_s$.
\begin{equation}\label{OriginalRSIG}
  Ring_{sig}=(pk_1,pk_2,...,pk_n,v,x_1,x_2,...,x_n),
\end{equation}
\par The signature verification by the smart contract as the fusion centre involves the following three steps.
\begin{itemize}
  \item[1.] Calculate all $y_i$ using corresponding public key $pk_i$ and $x_i$;
  \item[2.] Calculate the symmetric key $k = \mathcal{H}(msg)$;
  \item[3.] Verify that if $C_{k,v}(y_1,y_2,...,y_n)=v$ holds.
\end{itemize}
\par If the verification succeeds, the smart contract recognizes that $msg$ is sent by a valid sensing node from the group $\{Sensor_1, Sensor_2,..., Sensor_n\}$. In this way, we can cut the connection between the data packet and its corresponding owner.
\par  However, for the trust value update and payments assignment process, it is necessary to know the mapping of the sensing result and its corresponding sensor. To this end, we propose a two stage commitment scheme as follows.
\par This scheme consists of two stages including \textit{Commit Stage} and \textit{Reveal Stage}. 
As shown in Fig. \ref{Pic:PrivacyProtection}, at the commit stage, each sensor participating in the cooperative sensing needs to submit the $msg$ whose privacy is protected by the ring signature, and the field ``hash of $msgID$'', i.e. $\mathcal{H}(msgID)$, is used as a commitment. $msgID$ is an identifier which is the input of hash function, and can be any message such as "I am User 3" or "I like apples". 
At the reveal stage, each sensor directly uploads the unhashed $msgID$ for miners to verify the consistency. This checking mechanism is effective since hash function is
preimage resistance and second-preimage resistance \cite{Hash}, i.e., given a hash output $\mathcal{R}$=$\mathcal{H}(msgID)$, it is difficulty to find the input $msgID$ or another input $msgID'$ such that $\mathcal{H}(msgID)$ = $\mathcal{R}$ or $\mathcal{H}(msgID')=\mathcal{R}$. 

\begin{figure}[htbp]
\centering
	\includegraphics[width=3in]{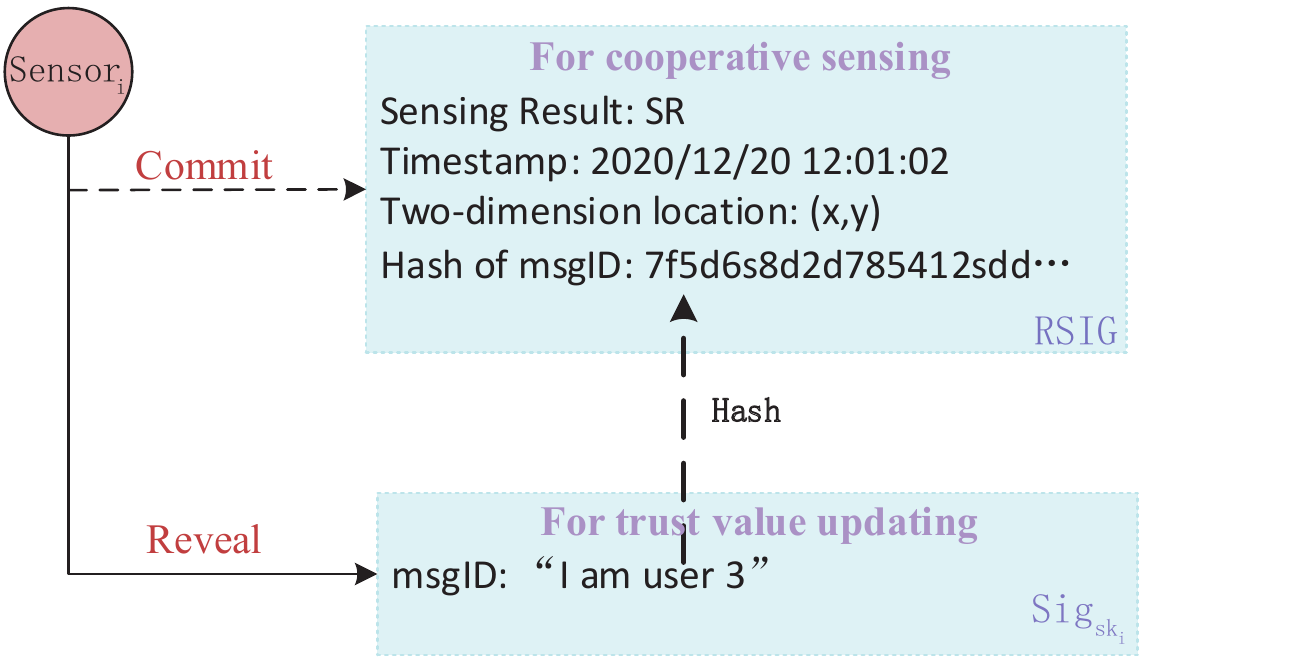}
\caption{Two-Stage commitment scheme.} \label{Pic:PrivacyProtection}
\end{figure}

\section{Smart Contracts and DSA Protocol Design} \label{sec:Protocol_and_implementation}
Smart contracts are computer codes that run on the blockchain platform, which are automatically executed when the predefined conditions are satisfied. Moreover, smart contracts are not controlled by any third party. Therefore, we propose to use smart contracts to realize the automated operation of cooperative spectrum sensing and spectrum auction on blockchain. In this section, we will give the designs of the corresponding smart contracts, and then propose the DSA protocol based on these smart contracts.

\subsection{Smart Contract Design}\label{subsec:SC}
We first discuss the design of our smart contracts for cooperative sensing and spectrum auction. In the following, we describe the parameters and functions in the two smart contracts, respectively.
\par \emph{Cooperative Sensing Contract (CSC)}: The parameters in CSC that need to be predefined are $T_{ddl}$, $d_s$, $N_1$ and $TV_{thr}$, which are introduced as follows.
\begin{itemize}
	\item  $T_{ddl}$: The deadline for a sensor to send its sensing packet to CSC; 
	\item $d_s$: The deposit that a sensor needs to pay before participating in sensing, which is used to be a guarantee for acting honestly. The deposit can be withdrawn only if the sensor uploads the sensing result which is consistent to the cooperative sensing result..
	\item $N_1$: The maximum number of sensors needed in cooperative sensing. If there are more than $N_1$ sensors who apply to participate in cooperative sensing, those with $N_1$ highest trust value will be selected.
	\item $TV_{thr}$: the minimum trust value of a sensor is needed to participate in cooperative sensing.
\end{itemize}
\par The functions in CSC are introduced as follows.
\begin{itemize}
	\item  $SensorRegister(\cdot)$: The inputs to this function are address $Addr$, deposit $dpt$ and trust value $TV$ of the node who invokes this function. This function will decide whether this node can register successfully by checking its trust value.
	\item $CheckRegisterQuality(\cdot)$: This function is used to check node's quality of participating in sensing by checking the parameters passed into it, i.e., node's trust value $TV_i$ and deposit $deposit_i$. For new users, participating in sensing is the only way to improve their trust value. This function offers users a way to convert their tokens to additional trust value, which enables the users whose trust value is below threshold to participate in sensing by paying more deposit.
	\item $Fusion(\cdot)$: The input to this function is $msgList$, which is a list consisting of all the sensing results from the registered sensors. This function outputs the final sensing result according to the pre-defined fusion rule.
	\item $UploadSensingData(\cdot)$: The inputs to this function are $msg$ and $RSIG$ , which are the sensing packets in (\ref{equ:msg}) and the ring signature of the node, respectively. This function is invoked by legal sensors to upload their sensing result. 
\end{itemize}
\par The pseudocode of the CSC is summarized in Algorithm \ref{contract:CSC}.

\par \emph{Spectrum Auction Contract (SAC)}: The parameters in SAC are as follows. 
\begin{itemize}
	\item $CSC_{id}$: The identity of corresponding CSC, which is used to identify the connection between SAC and CSC since every SAC is associated with a particular CSC. 
	\item $N_2$: The maximum number of the bidders, which is used to prevent too many nodes from sending bidding message.
	\item $T_{self-d}$: The time when this SAC will conduct the self-destruct operation to release memory.
	\item $d_a$: The amount of tokens that the bidders need to deposit before the final sensing result is released.
\end{itemize}
\par The functions in SAC are introduced as follows. 
\begin{itemize}
	\item  $BidderRegister(\cdot)$: Bidders invoke this function to register in SAC.
	\item $Commit(\cdot)$: The input to this function is $bldBid$ which denotes the commitment of $bidder_i$. Note that since the account balance is transparent, nodes can identify others' bid by checking their balance\cite{BlindAuction}. Therefore, everyone makes several commits in order to prevent others from inferring your bid based on changes of their account balance.
	\item $Reveal(\cdot)$: The inputs to this function include $Dps$, $Bools$ and $RNDs$. $Dps$ is the list of deposit ($Dp$) the bidders make; $Bools$ is the list of boolean values which indicate whether these bids are valid or not; $RNDs$ are random values to make the hash of  $\{Dp$, $Bool$, $RND\}$ hard to guess. A $Commit$ is the hash of these three parameters, i.e., $Commit$= $\mathcal{H}$$\{$$Dp$, $Bool$, $RND\}$. $bidsList[msg.sender]$ is the list of the invoker's commits. By comparing the reveal messages and the commits, this function can identify which bids are revealed correctly. Valid bids are added together as the bidder's total bid, invalid but correctly revealed bids are allowed to be withdrawn, bids that are not correctly revealed will not be returned.
	\item $Win(\cdot)$: The input to this function is $bidderList$, which is the list of all valid bids. Here we adopt the second-price
	sealed-bid auction for better economic profit \cite{Liang2020InBook}.  Accordingly, this function selects the bidder with the highest bid as the winner, who only needs to pay the second highest bid. 
	\item $EndOfAuction(\cdot)$: SAC will execute self-destruct operation at the preset time $T_{self-d}$.
\end{itemize}
\par The pseudocode of SAC is shown in Algorithm \ref{contract:SAC}.

\begin{algorithm}[tb]
	\caption{Cooperative Sensing Contract}\label{contract:CSC}
	\begin{algorithmic}[1] 
		\Function {init}{$TV_{thr}$, $d_s$, $N_1$} 
		\State require(msg.sender==ContractOwner);
		\State initialize $TV_{thr}$, $d_s$, $N_1$;
		\EndFunction
		\State
		\Function {SensorRegister}{$Addr, dpt, TV$}
		\State $sensor_{num} \gets 0$;
		\If {\Call{checkRegisterQuality}{$TV_i$, $addr_i$}}
		\State $sensorMap \gets sensor_i$;
		\State $sensor_{num} \gets sensor_{num}+1$;
		\State EMIT event("Registration success!");
		\Else
		\State EMIT event("Registration failed, please checkout the trust value and deposit.");
		\EndIf
		\If{$sensor_{num}>N_1$}
		\State Select top $N_1$ sensors according to trust value;
		\EndIf
		\EndFunction
		
		\State
		\Function{Fusion}{$msgList$}
		\State  Decision fusion in majority rule;
		\State \Return Cooperative sensing result;
		\EndFunction
		\State
		\Function{checkRegisterQuality}{$TV_i$,$deposit_i$}
		\If{ $deposit_i>d_s$}
		\State $TV^{'}_i$ $\gets$ $convert(deposit_i-d_s)$
		\If {($TV_i+TV^{'}_i)>TV_{thr}$ }
		\If {$totalNum<N_1$ \textbf{or} $TV_i>lowest TV$}
		\State Register Successful;
		\EndIf
		\Else
		\State Register Failed;
		\EndIf
		\EndIf
		\EndFunction
		\State
		\Function{UploadSensingData}{$msg$, $RSIG$}
		\If {$RSIG$ is legal}
		\State $msgList \gets msg$;
		\State EMIT event("Upload successfully!");
		\Else
		\State EMIT event("Upload failed, illegal sensor!");
		\EndIf
		\EndFunction
	\end{algorithmic}
\end{algorithm}

\begin{algorithm}[tb]
	\caption{Sealed Spectrum Auction Contract}\label{contract:SAC}
	\begin{algorithmic}[1] 
		\State $Bid \gets {bldBid_i, msg.value}$;
		\State $bidsList \gets mapping(address => Bid[]) $;
		\State $refund \gets 0$
		\State
		\Function {Commit}{$bldBid$}
		\State $bidsList[msg.sender].push(bldBid);$
		\EndFunction
		\State
		\Function {Reveal}{$Dps,Bools, RNDs$}
		\State $len \gets bidsList[msg.sender].length$;
		
		\While {$len>0$}
		\State $Commit \gets \mathcal{H}\{Dps[len], Bools[len], RNDs[len]\}$;
		\If{$Commit!=bidsList[msg.sender][len]$}
		\State EMIT event("Illegal Reveal!");
		\State Continue;
		\ElsIf{$Commit==bidsList[msg.sender][len]$ \textbf{and} $Bools[len]==true$} 
		\State $refund \gets refund + Dps[len]$;
		\ElsIf{$Commit==bidsList[msg.sender][len]$ \textbf{and} $Bools[len]==false$}
		\State withdraw invalid but reveal correctly bid;
		\EndIf
		\State $len \gets len-1$;
		\EndWhile
		\EndFunction
		
		\State
		\Function{Win}{$bidderList$}
		\State \Return Bidder with the highest bid;
		\EndFunction
		\State
		
		\Function{EndofAuction}{$T_{self-d}$}
		\State Execute self-destruct operation;
		\EndFunction
	\end{algorithmic}
\end{algorithm}

\subsection{The DSA Protocol} \label{subsec:Protocol}
In this part, we illustrate our proposed blockchain-based DSA protocol. To make it clearer, we elaborate this protocol in Fig. \ref{workflow}.
\begin{itemize}
	\item Phase 1 (\emph{Spectrum Sensing Request}): To ask for the channel state, SUs need to send a request message to Task Issuer (TI) through the control channel. TI is played by the node whose trust value is currently the highest. The role of TI is set to prevent multiple contracts from appearing in the network, which results in users participate in different contracts. TI will then creates and deploys the CSC and the corresponding SAC in the blockchain.
	\item Phase 2 (\emph{Smart Contracts and Nodes Register}): The CSC and SAC are instantiated when TI issues the transaction
	\begin{equation}\label{STC}
		CSC=\{CSC_{id} \mid T_{ddl} \mid N_1 \mid TV_{thr} \mid Fusion(\cdot) \mid d_s \}
	\end{equation} with the signature $Sig_{TI}(STC)$, and the transaction
	\begin{equation}\label{AC}
	SAC=\{CSC_{id} \mid SAC_{id} \mid  N_{2}  \mid T_{self-d} \mid Win(\cdot) \mid d_{a} \}
	\end{equation} 
	with the signature $Sig_{TI}$=$(AC)$. To register to CSC, a sensor issues the transaction
	\begin{equation}\label{Deposit}
	  Deposit_{CSC_{id}}=\{pk_i \mid TV_{i} \mid CSC_{id} \mid d_s\}
	\end{equation}  
	with the signature $Sig_{sk_i}(Deposit_{CSC_{id}})$ to make deposits. Similarly, a SU who intends to participate in SAC makes its deposit by issuing
	\begin{equation}\label{Bty}
	  Bty_{SAC_{id}} = \{pk_i \mid TV_{i} \mid SAC_{id} \mid d_a \mid \mathcal{H}(bid_i,RND_i)\}
	\end{equation} 
	with the signature $Sig_{sk_i}(Bty_{SAC_{id}})$. The $\mathcal{H}(bid_i, RND_i)$ denotes the bidding commitment this SU makes for the sealed auction.
	\item Phase 3 (\emph{Players selection}): Players (including sensors and bidder) are selected based on $SensorRegister(\cdot)$ and $BidderRegister(\cdot)$, respectively. Finally, an event will be triggered to provide an appropriate notification to every eventually selected nodes.
	\item Phase 4 (\emph{Sensing phase}): Sensors upload the sensing data by issuing the transaction
	\begin{equation}\label{SR}
	  msg_{tx}=\{CSC_{id} \mid \mathcal{H}(msgID) \mid TS \mid SR \mid loc\}
	\end{equation}  
	with the ring signature  generated by (\ref{OriginalRSIG}) before $T_{ddl}$, and $TS$ denotes timestamp, $SR$ denotes the sensing result. In this process, we use the ring signature to protect the SU’s identity. Meanwhile, these sensors need to make commitments mentioned in Section \ref{Subsec:Privacy protection} by sending the transaction 
	\begin{equation}\label{Commit}
   		Commit_i =\{CSC_{id}\mid \mathcal{H}(SR, RND_{i}, msgID) \mid pk_i)\
	\end{equation}  
	with $Sig_{sk_i}(Commit_i)$. After the sensing deadline, data fusion is done, and the final sensing result will be published in the blockchain network.
	\item Phase 5 (\emph{Auction if necessary}): If the cooperative sensing result indicates that the corresponding spectrum bands are idle, the SAC will be invoked. Users who have been selected will reveal their bids for the spectrum by sending bids and random numbers. Then the winner will be published on the blockchain.
	\item Phase 6 (\emph{Trust value updating}): In this phase, the trust value of each node is updated using equation (\ref{con:tvc2}) by miners in the network.

\end{itemize}

\begin{figure}[htbp]
	\centering
	\includegraphics[width=3.2in]{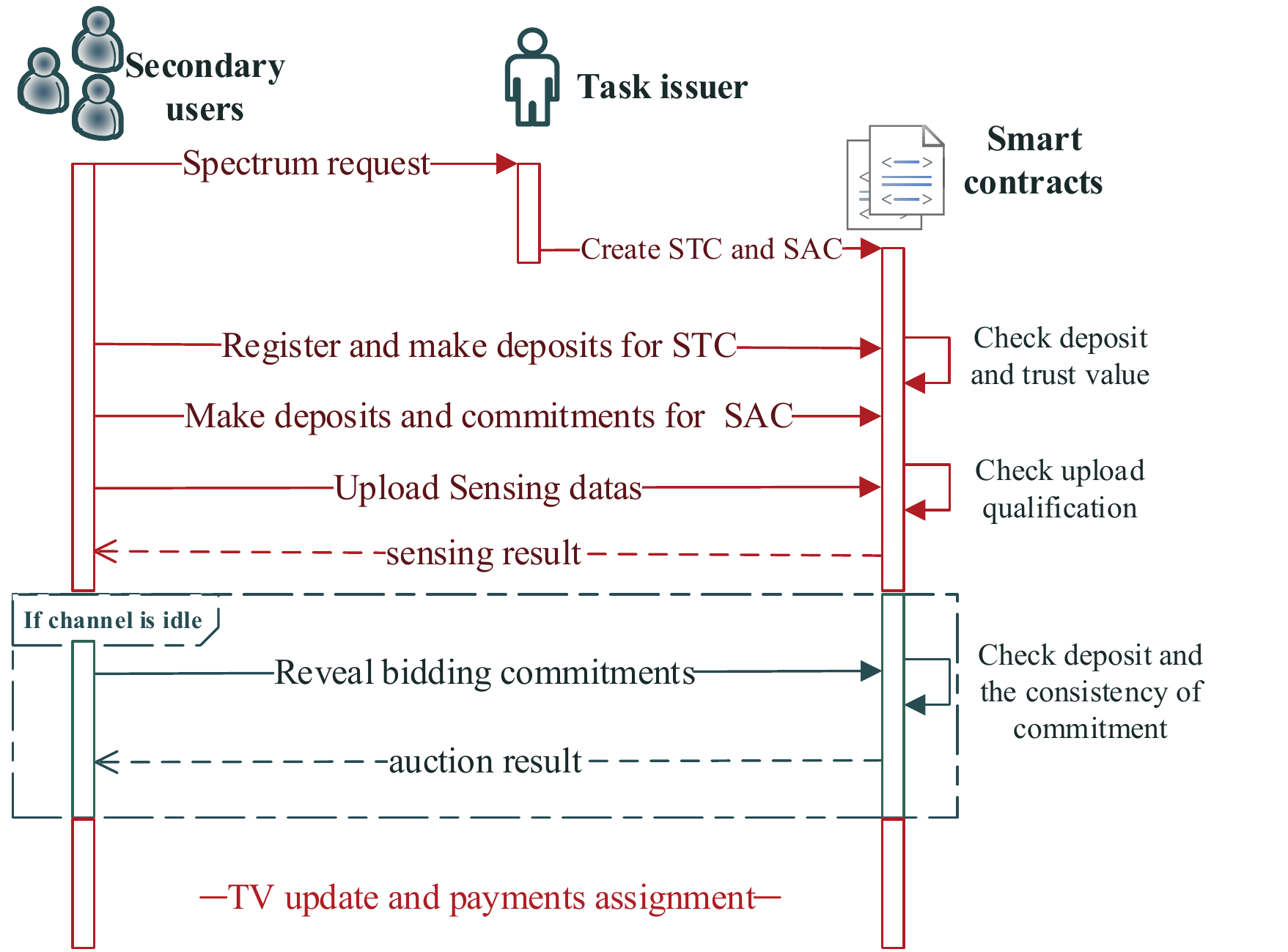}
	\caption{The workflow of our proposed blockchain-based DSA.} \label{workflow}
\end{figure}

\section{Implementation and Performance Analysis}\label{sec:System analysis}

\subsection{Implementation of Smart Contract}\label{subsec:Implementation}
We implement the proposed smart contracts using Solidity\cite{Solidity} in the Ethereum Virtual Machine (EVM). The smart contracts are compiled and deployed in the Remix IDE\cite{Remix} which is used for testing. After that, the nodes in the blockchain network can send transactions to deploy the smart contracts and invoke functions in these contracts. 
In the cooperative sensing phase, five sensors are considered and the information of their accounts in the blockchain network are listed in table \ref{accounts}. In the following, we describe the implementation procedures in detail.
\begin{enumerate}
	\item \emph{Smart Contract Preparation}: The smart contract is created and deployed by the node with the highest trust value, i.e., the fifth account in table \ref{accounts}. The relevant parameters and functions which need to be predefined are set as follows: $N_1$ is set as 3, the trust value threshold $TV_{thr}$ is set as 900 where we magnify the trust values by 100 times and make they as integers, since fixed point numbers are not fully supported by Solidity yet. $d_s$ is set as 100 wei,
	\item \emph{Smart Contract Deployment}: Fig.\ref{Contract_initialized} shows the details about the deployed smart contract. The field ``from" is the account address of TI. The field ``decoded input" shows all of our preset parameters. 
	\item \emph{Sensor Registration}: All the five sensors send their registration request to the smart contract. Then, three sensors are selected by the function $SensorRegister(\cdot)$. 
	\item \emph{Sensing Results Fusion}: The function $Fusion(\cdot)$, where the majority rule is selected, is used to fuse the sensing results from three legal sensors. 
\end{enumerate}
\par If the spectrum is detected as idle, the auction phase initiates. In this phase, we consider two bidders: each bidder has one true bid and one false bid. The related information is listed in Table \ref{table:bidder}. The first bidder has two bids: one is a true bid with 100 wei, the other one is a false bid with 200 wei. The second bidder also has two bids: one is true with 150 wei, the other is false with 300 wei.

\begin{table}[htbp]
	\centering  %
	\caption{Related information of sensor accounts}  %
	\resizebox{\columnwidth}{!}{
		\label{accounts}  %
		\begin{tabular}{|c|c|c|}
			\hline  %
			& \\[-6pt]  %
			Account addresses & Trust value & Sensing Result \\  %
			\hline
			& \\[-6pt]  %
			0x5B38Da6a701c568545dCfcB03FcB875f56beddC4 &  0.91 & 0 \\ 
			\hline
			0xAb8483F64d9C6d1EcF9b849Ae677dD3315835cb2 &  0.92 & 1 \\
			\hline
			0x4B20993Bc481177ec7E8f571ceCaE8A9e22C02db &  0.87 & 1 \\
			\hline
			0x78731D3Ca6b7E34aC0F824c42a7cC18A495cabaB &  0.93 & 0 \\
			\hline
			0x5B38Da6a701c568545dCfcB03FcB875f56beddC4 &  0.94 & 1 \\
			\hline
	\end{tabular}}
\end{table}

\begin{figure}[htbp]
	\centering
	\fbox{\includegraphics[width=3.4in]{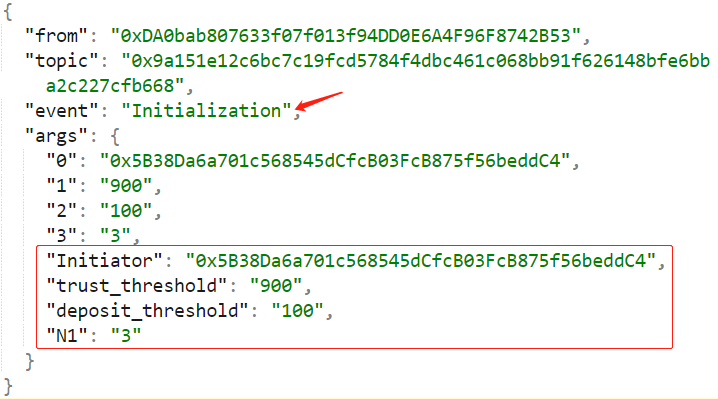}}
	\caption{Initialization of Smart Contract}\label{Contract_initialized}
\end{figure}

\begin{enumerate}
	\item \emph{Commit}: Bidders invoke $Commit(\cdot)$ function to make commits. Fig. \ref{Pic:BiddingCommit} shows the log of commit made by bidder 2. There are address information and bid information about this commit, thus others can know that bidder 2 make a bid with 150 wei, but they can not identify whether this is a real commit or not.
	\item \emph{Reveal}: After the commit phase, bidders invoke the function $Reveal(\cdot)$ to reveal its commits. The field ``decoded input'' in Fig. \ref{Pic:BiddingReveal} shows the whole information about this commit. The smart contract will automatically verify the information uploaded in the two phase and the final bidding result can be accessed by all nodes in the system.
\end{enumerate}


\begin{table}[htbp]
	\centering  %
	\caption{Related Information of Bidders}  %
	\resizebox{\columnwidth}{!}{
		\label{table:bidder}  %

		\begin{tabular}{|c|c|c|c|}
			\hline  %
			& \\[-6pt]  %
			Bidder address & Bid & isReal & Commit\\  %
			\hline
			& \\[-6pt]  %
			\multirow{2}{*}{0xAb84...5835cb2} &  100wei & yes & 0x591a291ad67...1b62c96a2092 \\ 
			\cline{2-4}
			~ &  200wei & no & 0x4871a30b671...4d08cfaa8bc1 \\ 
			\hline
			\multirow{2}{*}{0x4B20...2C02db} &  150wei & yes & 0x4e6a24e35c7be...23643d68efa \\
			\cline{2-4}
			~ &  300wei & no & 0x0c98edecac...74d24d971ec88 \\
			\hline
	\end{tabular}}
\end{table}	

\begin{figure}[htbp]
	\centering
	\fbox{\includegraphics[width=3.4in]{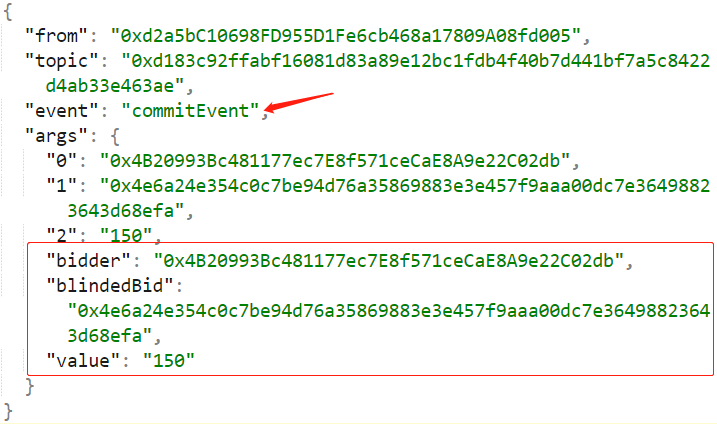}}
	\caption{The commit made by bidder 2} \label{Pic:BiddingCommit}
\end{figure}

\begin{figure}[htbp]
	\centering
	\fbox{\includegraphics[width=3.4in]{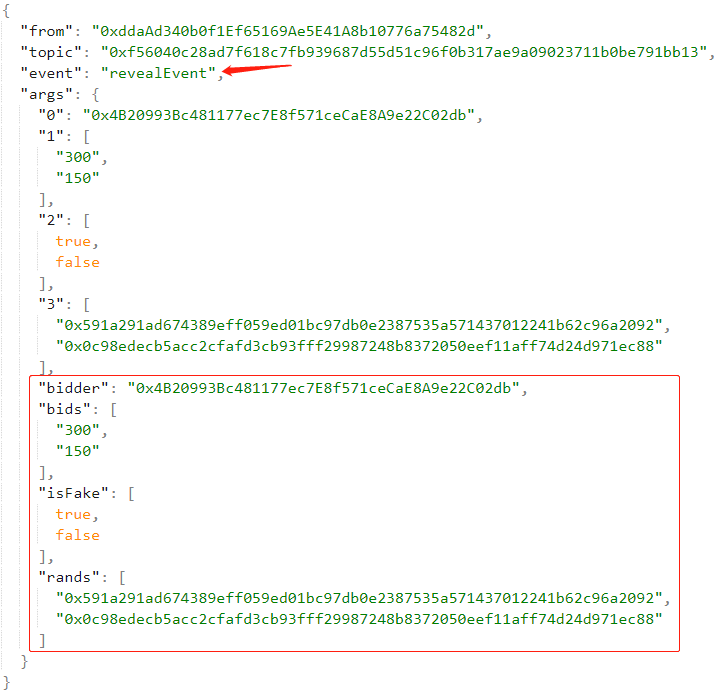}}
	\caption{The reveal made by bidder 2.} \label{Pic:BiddingReveal}
\end{figure}

\subsection{Performance Analysis of The Proposed PoT Consensus Algorithm}\label{subsec:tvPoW}
In this part, we evaluate the performance of the proposed PoT consensus algorithm. We first discuss how to set the initial mining difficulty at the beginning when every node's trust value is 0. We simulate the mining process in our personal computer as a reference to explain how to select a suitable mining difficulty. We use a given string to represent the transactions in a block in the blockchain network. The result is shown in Fig. \ref{Difficulty}. The running time increases exponentially when the number of leading zeros is larger than 20. Supposing that the hash rate of the IoT devices is similar to that of our personal computer, the interval between two blocks $T_0$ is about 1 second, then the mining difficulty should be less than 18 leading zeros, the $\beta_{0}$ is about $2^{18}=262144$.
\par It’s worth mentioning that even under the same mining difficulty and experimental environment, the mining time may be different. For example, given the mining difficulty where 32 leading zeros are needed in the target hash, it sometimes take tens of seconds, and sometimes can take more than an hour. In order to eliminate the impact of luck, we convert the mining difficulty into the \textit{expected mining cost} and evaluate the expected mining cost instead of running time. Theoretically, assuming that the output of each hash calculation is unpredictable, each time when we increase the leading zero by one, the success probability of each hash trial will be halved. 
The trust value is associated with the mining difficulty. Therefore, we first convert the trust value into the number of leading zeros, and then the number of leading zeros is converted into the expected mining cost proportionally. Simulation is conducted for 1000 time slots, and the \textit{average expected mining cost} of every type of nodes is calculated in Fig.\ref{TVPoW}. 
To be specific, we simulate a network with 20 nodes and 4 types of nodes are considered.
\textit{Rnode} means a node with high performance and always act honestly; 
\textit{UAnode} means a node with high performance and act honestly, however it participates in sensing infrequently. 
\textit{OOnode} is a node who conducts malicious behavior periodically. We assume that \textit{OOnode} will act maliciously after every two normal sensing rounds. 
The \textit{Lnode} is a node who randomly provides binary sensing result without sensing. 
There are 12 \textit{Rnode} with probability of detection $p_d$ = 0.90 and probability of false alarm $p_f$ = 0.15 \cite{KangCR};
3 \textit{OOnode} whose $p_d$ = 0.90 and $p_f$ = 0.15; 
3 \textit{Lnode} with $p_d$ = 0.5 and $p_f$ = 0.5; and 
2 \textit{UAnode} with $p_d$ = 0.90 and $p_f$ = 0.15. Besides, it is assumed that they only have a 50\% chance to sign up in CSS each sensing round. 

\par In Fig.\ref{TVPoW}, it can be observed that the reliable node performs best with respect to the average expected mining cost per timeslot, which is about one third of the other three types of nodes. The expected mining cost of a on-off node is a little smaller than \textit{Lnode} and \textit{UAnode} but still cost three times than \textit{Rnode}. This proves the effectiveness of our proposed consensus algorithm.

\begin{figure}[tb]
	\centering
	\includegraphics[width=3.5in]{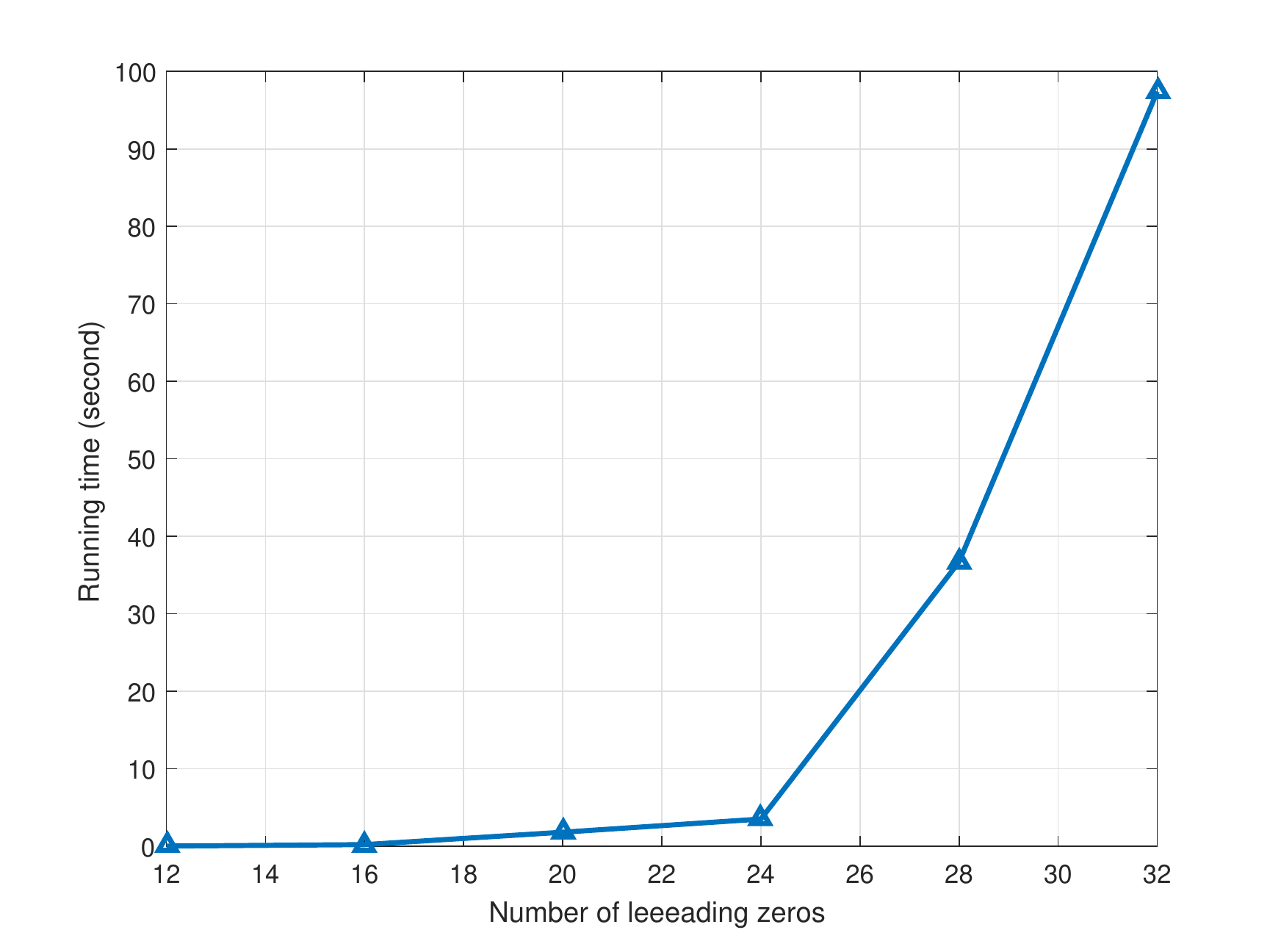}
	\caption{Running time (second) with increasing number of leading zeros} \label{Difficulty}
\end{figure}

\begin{figure}[tb]
	\centering
	\includegraphics[width=3.5in]{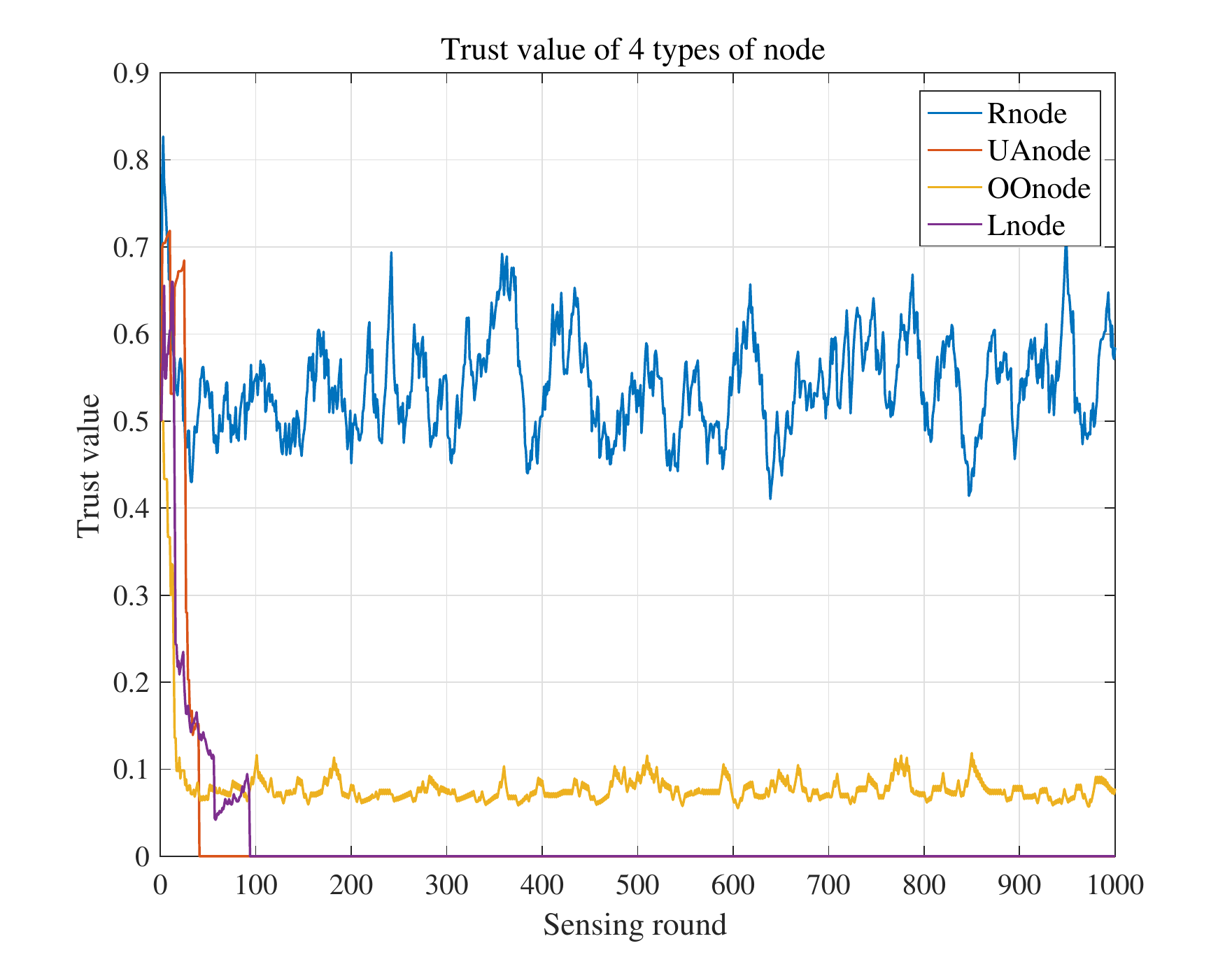}
	\caption{Curve of trust value}\label{TViPd}
\end{figure}

\begin{figure}[htbp]
\centering
\includegraphics[width=3.5in]{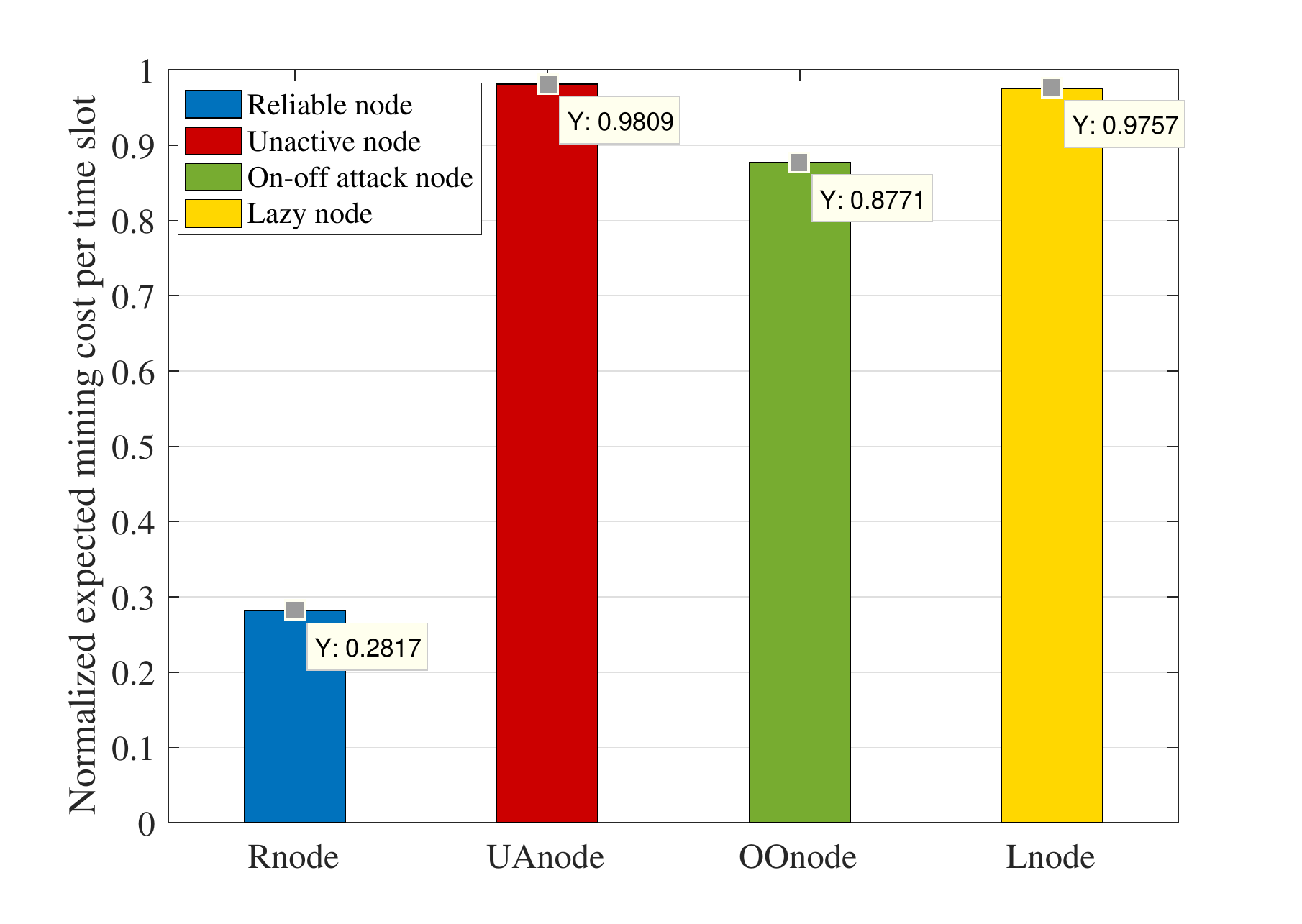}
\caption{Expected mining consumption of different nodes} \label{TVPoW}
\end{figure}

\subsection{Performance of Cooperative Sensing}\label{Subsec:Security analysis}
In this part, we evaluate the performance of cooperative spectrum sensing with the help of our proposed trust evaluation mechanism. In traditional blockchain network, there is no trust among nodes hence anyone can record the sensing information into blockchain, including nodes with poor performance or malicious behavior. However, due to the introduction of the trust value mechanism, the cooperative sensing contract can exclude bad nodes according to nodes' trust value, thus the system's sensing performance can be improved. 
For comparison, we consider 3 kinds of selection schemes to select the candidate nodes for the cooperative sensing. (i) random selection scheme; (ii) select according to register time; (iii) select according to sensor's trust value. We simulate a 20 nodes network with the same setup in \ref{subsec:tvPoW}. 

It can be seen from Fig.\ref{CoSensingPerformance} that when the number of needed sensors $N_1$ is small, $p_d$ and $p_f$ of cooperative sensing in the first two schemes which do not take advantage of trust value perform worse than that of the last selection scheme. This indicates that our proposed trust value mechanism can effectively improve the cooperative sensing performance of the system.
Moreover, in the last selection scheme, $p_d$ is close to $1$ when $N_1$ is about one fourth of all network nodes. The fewer nodes involved, the smaller the total monetary reward that the network needs to pay to sensors, and the smaller the economic burden for spectrum buyers.
As $N_1$ grows, the performance of different schemes will gradually close, since most of the network nodes (including bad nodes) will participate in cooperative sensing.

\begin{figure*}[htbp]
\centering
\subfigure[$P_d$ of cooperative sensing]{ 
	\includegraphics[width=3.4in]{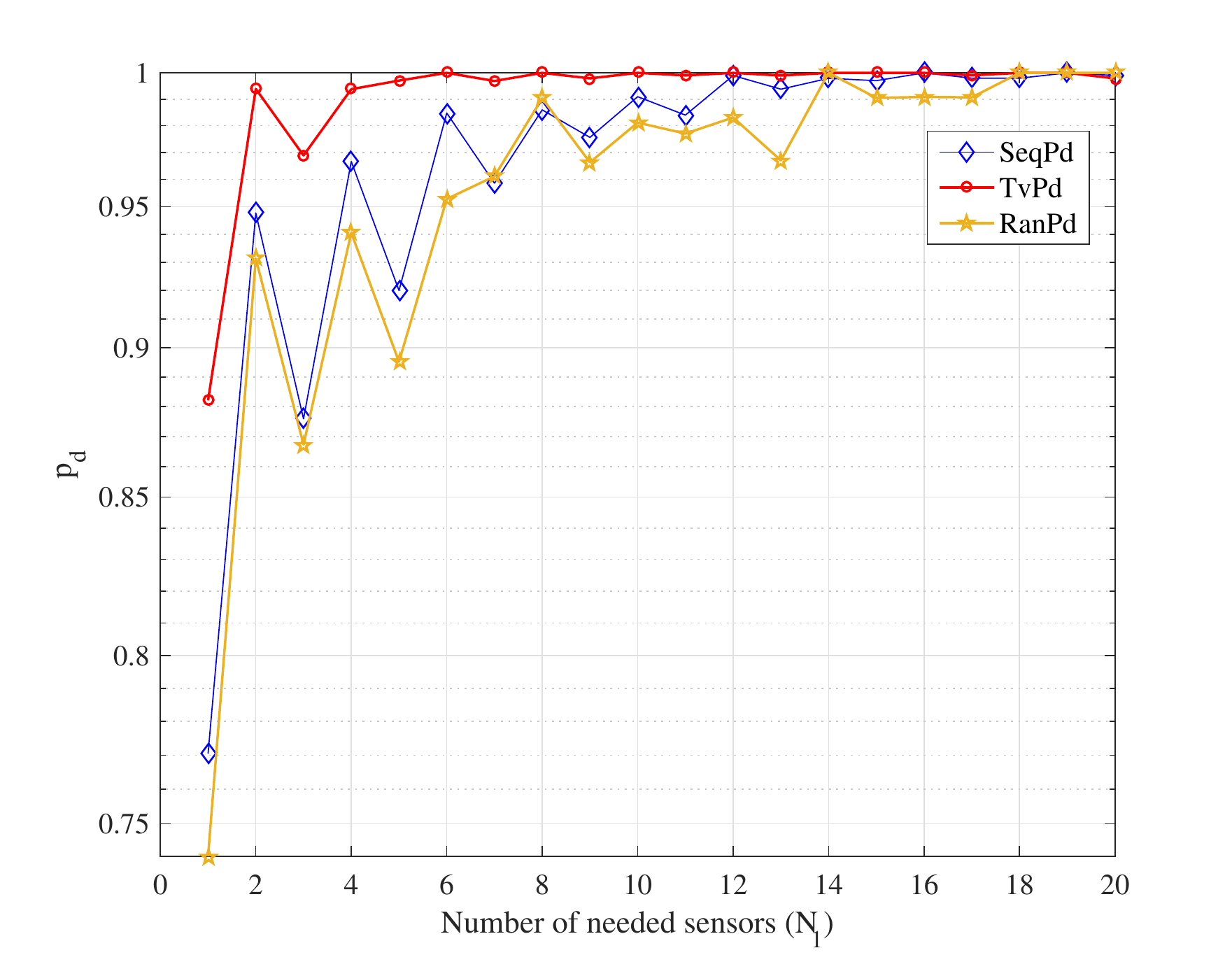} }
 \subfigure[$P_f$ of cooperative sensing]{ 
 	\includegraphics[width=3.42in]{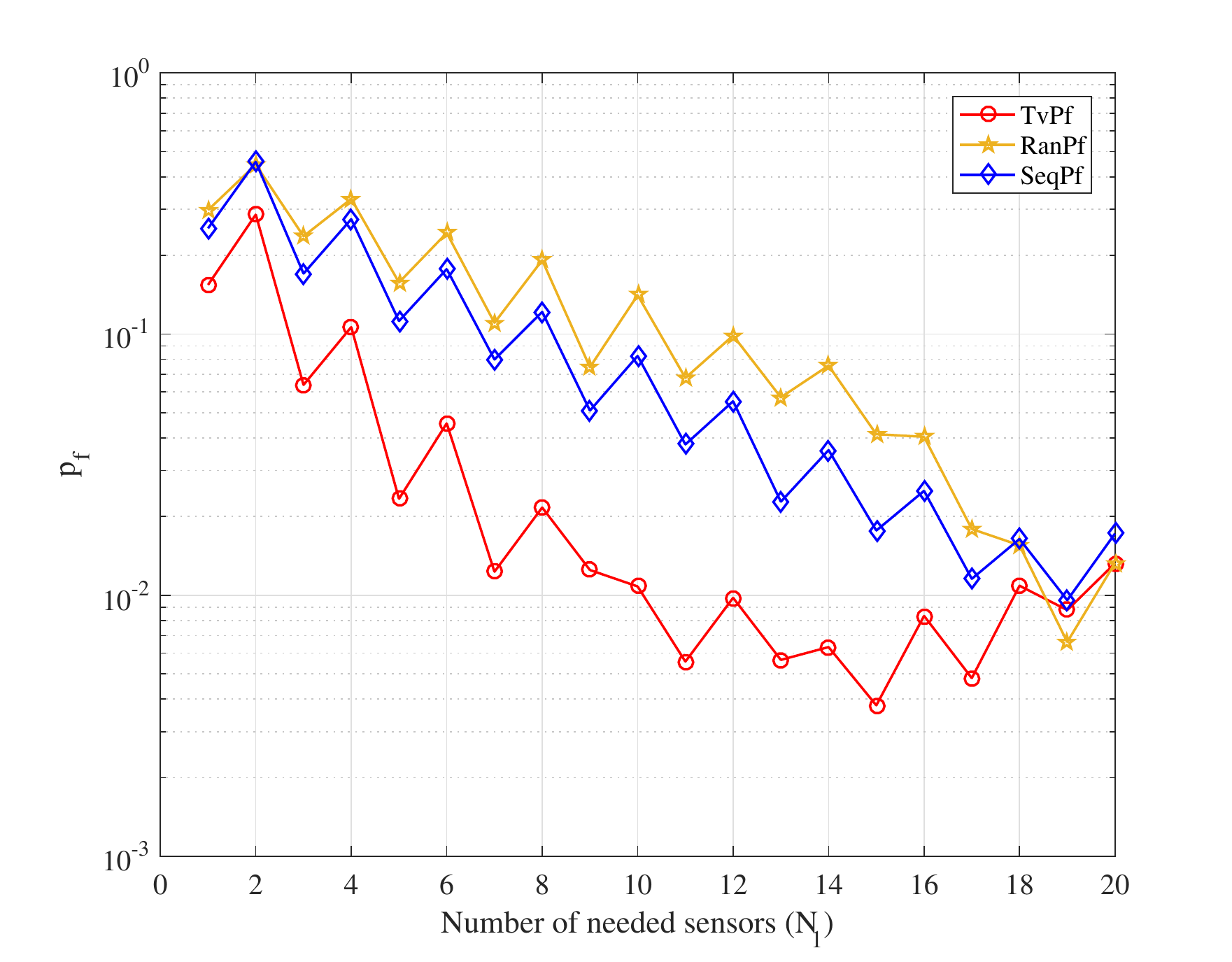} } 
 \caption{Performance of cooperative sensing under three selection schemes} \label{CoSensingPerformance}
\end{figure*}

\subsection{Incentive Mechanism Analysis}\label{Subsec:Incentive mechanism}
\par Since both sensing and mining will consume SU's computing power, an incentive mechanism is needed to reward nodes for the work they have done. The users who upload a sensing result that is consistent to the final cooperative sensing result will be rewarded with $R_s$ tokens; The users who successfully mine will be rewarded with $R_m$ tokens. The tokens can be used in auction to bid for spectrum resources. In implementation, the tokens rewarded for accurate spectrum sensing and successful mining, i.e., $R_s$ and $R_m$, need to be designed based on the evaluation of the computation consumption of spectrum sensing and mining.
\par This incentive mechanism not only encourages the nodes to participate in spectrum sensing, but also encourage them to behave honestly and accurately in spectrum sensing. This is because, on the one hand, an honest and accurate sensing node is more likely to derive a sensing result that is  consistent to the final cooperative sensing result and obtain tokens rewarded for spectrum sensing. On the other hand, with the proposed consensus algorithm, the node with a higher trust value is easier to mine successfully so that they have a higher probability of obtaining tokens for mining. On the contrary, the dishonest nodes will be discouraged since they are less likely to obtain rewards since they need to spend more computing resources on mining.
\subsection{Security Analysis}\label{sec:Security analysis}
\begin{itemize}
	\item \emph{Distributed Denial of Service (DDoS) Attack}: The DDoS attack here means that malicious users try to make the sensing service unavailable to other users. Our system is resist to this attack since we adopt the deposit mechanism. Thus, under our scheme, the cost of launching large-scale DDoS attack is very high since the attackers need to obtain lots of tokens in order to launch the attack. 
	\item \emph{Spoofing Attack}: Spoofing attack means someone tries to masquerade others to create forged transactions. Secure Elliptic Curve Digital Signature Algorithm (ECDSA)\cite{ECDSA} used in our blockchain can prevent this attack on the premise that attack does not have the user's private key.
	\item \emph{Free-riding Attack (lazy node)}: The Free-riding attack here means lazy users may directly copy others' sensing results at the phase of uploading sensing data. Firstly, there is no motivation for sensors to submit the sensing result before sensing deadline. Thus, when the lazy nodes get the sensing result, it is difficult for them to repack the sensing result and submit before the deadline. Secondly, even if a few sensors submit sensing results in advance, the connection between a user's identity and sensing data is cut by the ring signature, thus the lazy users cannot determine the owner of the sensing data, and thus the credibility of the data cannot be guaranteed. Thus, it is no better than submit a sensing result randomly. 
	\item \emph{On-off Attack}: On-off attack means that a node performs malicious behaviors periodically. If a trust management mechanism satisfy the condition that the dropping rate of trust value is larger than its increasing rate, it is considered to be resist to the on-off attack \cite{KangTV}. However, the model in \cite{KangTV} may cause the trust value drop a lot even the misbehavior is unintentional, since the increasing rate is much smaller than dropping rate. In our proposed model, as show in Fig.\ref{On-Off}, the unintentional mistakes made by sensors can be compensated by making right sensing decision. When the parameters $\tau$ and $\eta$ satisfy $\rho \textgreater \frac{\eta}{1-e^{-\eta}}-\eta$, our proposed trust value management mechanism is considered to be resist to the on-off attack, the proof is given in Appendix \ref{Appendix}. 
\end{itemize}

\begin{figure*}[htbp]
	\centering
	\subfigure[Existing model]{ 
		\includegraphics[width=3.4in]{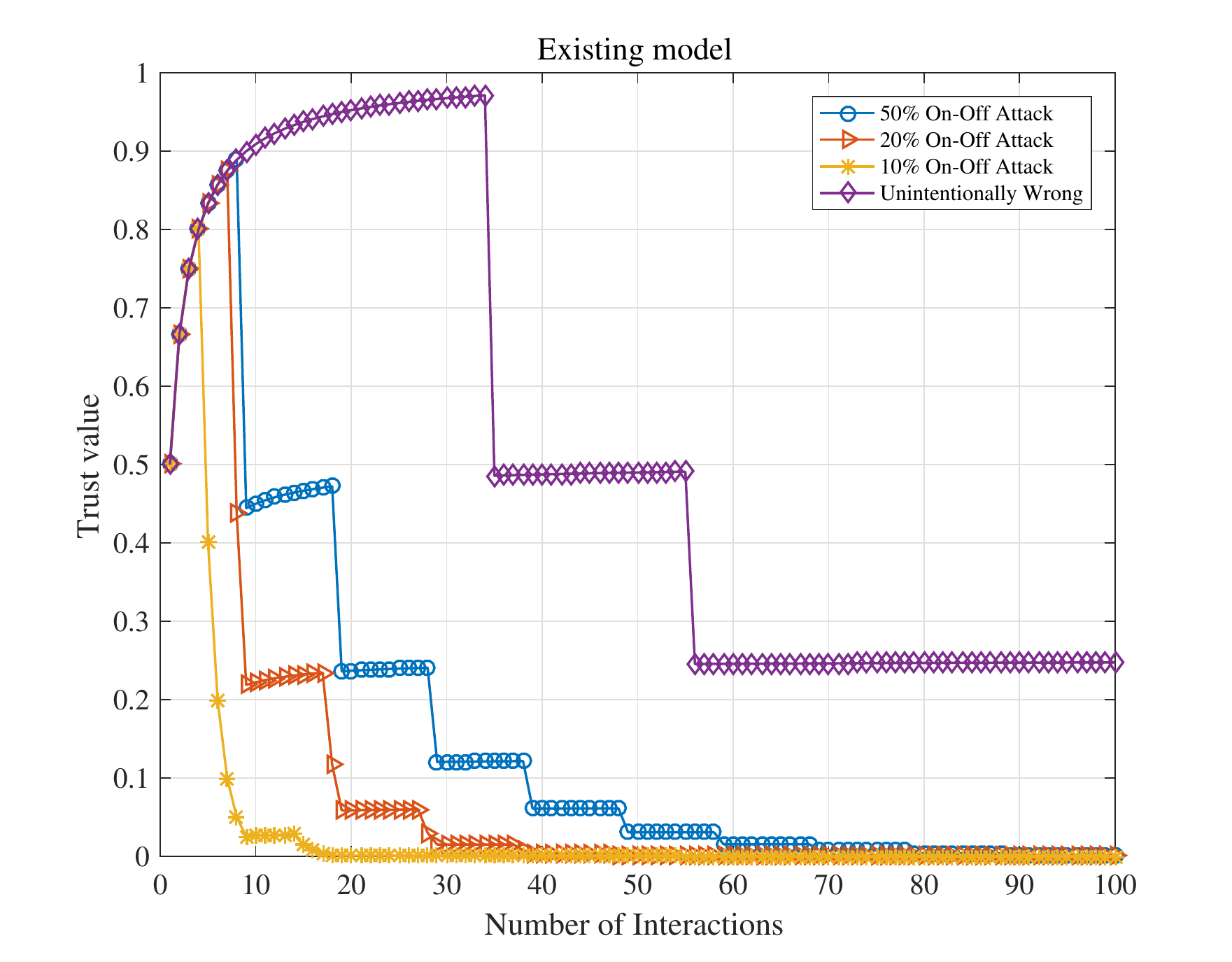} }
	\subfigure[Proposed model]{ 
		\includegraphics[width=3.4in]{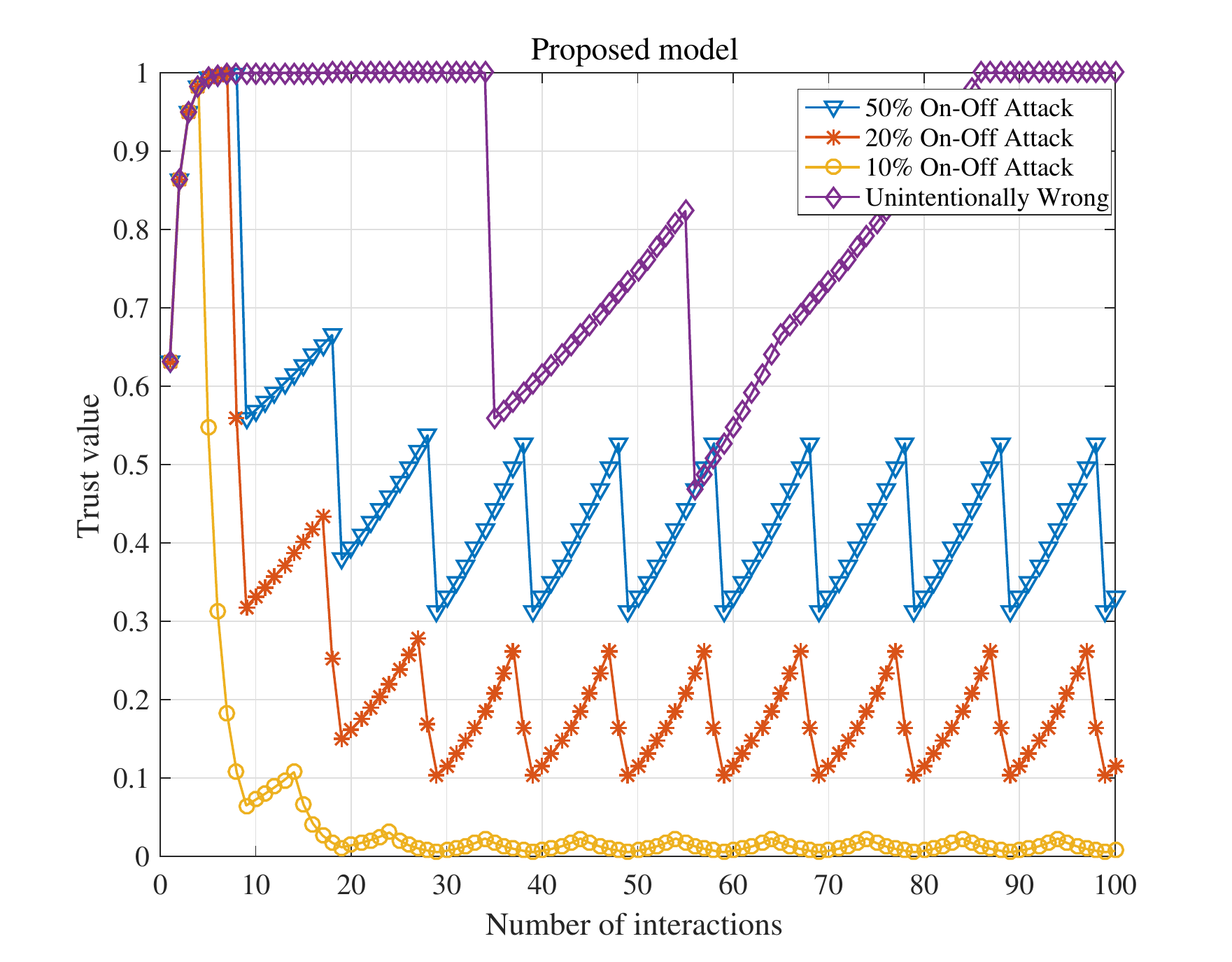} }
	\caption{Trust value of different types of nodes under two model} 
	\label{On-Off}
\end{figure*}

\section{Conclusion} \label{sec:conclusion}
In this paper, we have proposed a blockchain-based dynamic spectrum sharing protocol.  
This protocol mainly consists of three parts: the first part is the trust value management mechanism, which is designed to evaluate the credibility of nodes; the second part is the PoT consensus algorithm, which make the mining difficulty for malicious nodes greatly increased. The combination of node's trust value and their mining difficulty can motivate nodes to be more willing to behave honestly; the third part is the privacy protection mechanism, in which we combine the ring signature and the commit-reveal scheme to solve the problem of privacy issue in the process of cooperative spectrum sensing.
Finally, we implemented the prototype of our proposed smart contracts and analyzed the performance of PoT consensus algorithm and the improvement in cooperative sensing. Security analysis of the system show that our framework can resist many kinds of attacks which are frequently encountered in trust-based blockchain systems. 
\IEEEpeerreviewmaketitle

\appendices
\section{Proof of the resistance of on-off attack}\label{Appendix}
\begin{theorem}
The proposed trust model given in (\ref{con:tvc})
is resistant to on–off attack when $\rho \textgreater \frac{\eta}{1-e^{-\eta}}-\eta$.
\end{theorem}
\begin{proof}
	Let $f(N_r,N_w) = e^{-\rho\cdot N_w}\cdot \left(1-e^{-\eta \cdot N_r}\right)$ denotes the trust value update function in (\ref{con:tvc}). Because when $N_r=0$, $f(N_r,N_w)$ can only increase, this situation is not discussed here. At any other points $(N_r,N_w)$ in this function, the decreasing rate should be larger than the increasing rate, that is: 
	\begin{equation}
		\left\lVert \frac{\partial f}{\partial N_w} \right \rVert>\left\lVert \frac{\partial f}{\partial N_r} \right\rVert,
	\end{equation}
	where 
	\begin{equation}
	\left\lVert \frac{\partial f}{\partial N_w} \right \| = \rho \cdot e^{-\rho \cdot N_w}\cdot\left(1-e^{-\eta \cdot N_r}\right)
	\end{equation}
and 
	\begin{equation}
	\left \| \frac{\partial f}{\partial N_r} \right \| = \eta \cdot e^{-\rho \cdot N_w}\cdot \left(e^{-\eta \cdot N_r}\right),
	\end{equation}
	then we have 
	\begin{equation}
			\left \| \frac{\partial f}{\partial N_w} \right \|>  \left \| \frac{\partial f}{\partial N_r} \right \| \Leftrightarrow 
		\frac{\rho}{\eta}>\frac{e^{-\eta \cdot N_r}}{1-e^{-\eta \cdot N_r}}
	\end{equation}
	which can also be denoted as:
	\begin{equation}
		\left \| \frac{\partial f}{\partial N_w} \right \|> \left \| \frac{\partial f}{\partial N_r} \right \| \Leftrightarrow 
		\frac{\rho}{\eta}>\frac{1}{1-e^{-\eta \cdot N_r}}-1
	\end{equation}
	When $N_r \ge 1$, 
	\begin{equation}
		\frac{1}{1-e^{-\eta \cdot N_r}}-1 \le \frac{1}{1-e^{-\eta}}-1,
	\end{equation}
	Therefore, as long as it is satisfied $\frac{\rho}{\eta}>\frac{1}{1-e^{-\eta}}-1$, the above inequality holds.
\end{proof}








\end{document}